\documentclass[leqno]{article}

\usepackage[letterpaper]{geometry}

\usepackage[hidelinks]{hyperref}       
\usepackage{url}            
\usepackage{booktabs}       
\usepackage{nicefrac}       
\usepackage{microtype}      
\usepackage{graphicx}
\usepackage[numbers]{natbib}
\usepackage{amsmath,amsthm}
\usepackage{doi}
\usepackage{lmodern}
\usepackage{bbm}
\usepackage{amssymb}
\usepackage{doi}
\usepackage{tikz} \usepackage[symbol]{footmisc}

\usepackage{pgf,pgfplots,xspace}
\pgfplotsset{compat=1.14}
\usepgfplotslibrary{fillbetween}
\usepackage{hyphenat}
\usepackage{thm-restate}
\usepackage{algorithmic}
\usepackage[algo2e,ruled,noend,noline]{algorithm2e} 
\providecommand{\DontPrintSemicolon}{\dontprintsemicolon}
\usepackage{dsfont}
\usepackage{bm}
\usepackage{colortbl}
\usepackage{xcolor}
\usepackage{setspace}
\usepackage[misc]{ifsym}
\usepackage[normalem]{ulem}
\usepackage[capitalise,nameinlink]{cleveref}
\allowdisplaybreaks

\DeclareMathOperator*{\argmax}{arg\,max}

\newcommand{\E}[1]{\mathbb E \left[ #1 \right]}

\newcommand{\M}{M} 
\newcommand{\Ms}{M_S} 
\newcommand{\Gs}{G_S} 
\newcommand{\Eprime}{E'} 
\newcommand{\Eplus}{E^+} %
\newcommand{\Esafe}{E_{\text{safe}}} %

\newcommand{\Mat}{\mathcal{M}}
\newcommand{\opt}{\mathsf{OPT}}

\newcommand{\I}{\mathcal{I}}

\renewcommand{\O}{\mathcal{O}}

\newcommand{\EXP}[2][]{
    \ifthenelse{\equal{#1}{}}
    {\mathbb{E}\left[#2\right]}
    {\mathop{\mathbb{E}}_{#1}\left[#2\right]}
}
\newcommand{\PRO}[2][]{
    \ifthenelse{\equal{#1}{}}
    {\mathbb{P}\left(#2\right)}
    {\mathop{\mathbb{P}}_{#1}\left(#2\right)}
}
\newcommand{\sspi}{SSPI}
\newcommand{\psspi}{P-SSPI}
\newcommand{\oos}{OOS\xspace}
\newcommand{\1}[1]{\mathds{1}_{\left\{ #1 \right\}}}
\newcommand{\val}[2]{\mathrm{val}_{#1}(#2)}
\newcommand{\valEprime}[1]{\val{\Eprime}{#1}}
\newcommand{\valMs}[1]{\val{\Ms}{#1}}
\newcommand{\MM}{\mathcal{M}}

\newtheorem{theorem}{Theorem}[section]
\newtheorem{definition}{Definition}[section]

\newtheorem{claim}{Claim}[section]
\newtheorem{lemma}{Lemma}[section]
\crefname{claim}{Claim}{Claims}
\newtheorem{remark}{Remark}[section]
\newtheorem{observation}{Observation}[section]

\makeatletter
\let\original@algocf@latexcaption\algocf@latexcaption
\long\def\algocf@latexcaption#1[#2]{%
  \@ifundefined{NR@gettitle}{%
    \def\@currentlabelname{#2}%
  }{%
    \NR@gettitle{#2}%
  }%
\original@algocf@latexcaption{#1}[{#2}]%
}
\makeatother

\title{Single-Sample Prophet Inequalities via Greedy-Ordered Selection\thanks{An extended abstract of this work, which merges and extends two previous preprints \citep{CaramanisEtAl21,DFLLR21}, appeared in the Proceedings of the 33rd ACM SIAM Symposium on Discrete Algorithms \citep{CaramanisDFFLLP22}.}}

\author{
Constantine Caramanis\thanks{
The University of Texas at Austin, USA. Email: 
\texttt{\{constantine,matthewfaw\}@utexas.edu}}
\and
Paul D{\"u}tting\thanks{
Google Research,
Z\"urich, Switzerland.
Email: \texttt{duetting@google.com}}
\and
Matthew Faw$^{\dagger}$\and
Federico Fusco\thanks{
Sapienza University of Rome, Italy.
Email: \texttt{\{fuscof,leonardi\}@diag.uniroma1.it}}
\and
Philip Lazos\thanks{
IOHK. Email: \texttt{philip.lazos@iohk.io}}
\and Stefano Leonardi$^\S$\and
Orestis Papadigenopoulos 
\and 
Emmanouil Pountourakis\thanks{
Drexel University, Philadelphia, USA.
Email: \texttt{manolis@drexel.edu}}\and
Rebecca Reiffenh{\"a}user\thanks{
University of Amsterdam, Amsterdam, The Netherlands.
Email: \texttt{r.e.m.reiffenhauser@uva.nl}}
}

\date{March 15, 2024}

\begin{document}

\maketitle

\begin{center}
    \textit{We dedicate this paper to the memory of Orestis, who left us far too early.}
\end{center}

\begin{abstract}
We study \emph{single-sample prophet inequalities} (SSPIs), i.e.,  prophet inequalities where only a single sample from each prior distribution is available. In addition to a direct and optimal SSPI for the basic single choice problem (Rubinstein et al., 2020), most existing SSPI results were obtained via an elegant, but inherently lossy reduction to order-oblivious secretary (OOS) policies (Azar et al., 2014). Motivated by this discrepancy, we develop an intuitive and versatile greedy-based technique that yields \sspi{}s \emph{directly} rather than through the reduction to \oos{}s. Our results can be seen as generalizing and unifying a number of existing results in the area of prophet and secretary problems. Our algorithms significantly improve on the competitive guarantees for a number of interesting scenarios (including general matching with edge arrivals, bipartite matching with vertex arrivals, and certain matroids), and capture new settings (such as budget additive combinatorial auctions). Complementing our algorithmic results, we also consider mechanism design variants. Finally, we analyze the power and limitations of different SSPI approaches by providing a partial converse to the reduction from \sspi{} to \oos{} given by Azar et al.

\hspace{10 pt}


\end{abstract}

\clearpage

\section{Introduction}

Prophet inequalities are fundamental models from optimal stopping theory. 
In the simplest version of prophet inequalities \citep{KS77,KrengelS78,SamuelCahn84}, a set of $n$ elements with rewards $r_1, r_2, \ldots, r_n$ sampled independently from known distributions $D_1, D_2, \ldots, D_n$ 
arrive in a fixed order. After each element is observed, the online algorithm can either accept the reward and stop, or drop it and move to consider the next element in the sequence. 
The goal is to show an $\alpha$-competitive \emph{prophet inequality}, i.e., an online algorithm such that the expected reward of the computed solution is at least a $1/\alpha$-fraction of the expected value of an optimal solution that can observe the values of all rewards before making a choice. The original work of \citet{KS77,KrengelS78} and \citet{SamuelCahn84} has established a tight $2$-competitive prophet inequality for the single-choice selection problem, i.e.,  the problem of selecting only one out of $n$ elements. Motivated also by applications to algorithmic mechanism design and to pricing goods in markets, later work has extended prophet inequalities to a broad range of problems, including matroids and polymatroids, bipartite and non-bipartite matching,  and combinatorial auctions \citep[e.g.,][]{Alaei11,ChawlaHMS10, KleinbergW12,DuttingK15,FeldmanGL15,DuttingFKL20,DuttingKL20}. 

For many applications, the assumption of knowing the distributions exactly, as in the prophet model, is rather strong. It is consequential to ask whether near-optimal or constant-factor prophet inequalities are also achievable with limited information about the distributions. 
A natural approach in this context (pioneered by \citet{AzarKW14}) is to assume that the online algorithm has access to a limited number of \emph{samples} from the underlying distributions, e.g., from historical data. 
Arguably, a minimal assumption under this approach is that the online algorithm has access to only a \emph{single} sample from each distribution. Surprisingly, even under this very restrictive assumption, strong positive results are possible. Importantly, for the classic single-choice problem one can obtain a $2$-approximate prophet inequality with just a single sample \citep{RubinsteinWW20}, matching the best possible guarantee with full knowledge of the different distributions.
The same problem, but with identical distributions, was considered by \citet{correa2019prophet}, who showed that with one sample from each distribution, it is possible to achieve a $\nicefrac e{(e-1)}$ guarantee. 
This bound was subsequently improved  by Kaplan~et~al.~\cite{KaplanNR20}, and by Correa and coauthors ~\cite{CorreaCES20,CorreaDFSZ21,correa2020sample}. 

The central question we study in our work is whether similar results can also be obtained for more general combinatorial settings. This question was first investigated by \citet{AzarKW14}, 
who showed that for some combinatorial problems it is possible to achieve the same asymptotic guarantees as with full knowledge of the distributions with just a single sample from each distribution. 
The results they give are largely based on an elegant reduction to \emph{order-oblivious secretary} (OOS) algorithms: in this model, there is no sample or other information about the distributions available; instead, a constant and random fraction of the input is \emph{observed} for statistical reasons. Then, only on the remainder of the instance, an adversarial-order online algorithm is run to select elements. The reduction of \sspi{}s to \oos{}s 
relies on the following observation: with the single sample available to an SSPI, one can replace the sampled part of the instance in the first, statistical phase of any OOS. Since the sample values are drawn from the exact same distributions as those that were part of the original problem instance, this results in the exact same input distribution as the original \oos{} (and, therefore, the same competitive guarantee).

However, \sspi{}s derived from \oos{}s in this fashion are inherently lossful: the information from the sample available to such an SSPI remains partly unused, since the secretary algorithm utilizes only those sampled values that belong to elements in its statistical, not its selection phase. 
Furthermore, the elements used in the statistical phase are rejected \emph{automatically}.
We deviate from designing SSPIs via this construction and instead give a versatile and intuitive technique to design SSPIs \emph{directly}.
Our focus is to design SSPIs directly for matchings, matroids, and combinatorial auctions that are central to the area of secretary and prophet algorithms.

Besides presenting improved or completely new \sspi{}s for many such problems, we also investigate the power and limitations of the \sspi{} paradigm with respect to prior-free \oos{}s on one side and full-information prophet inequalities on the other. Here, we take important steps towards solving the central open question whether constant-factor \sspi{}s exist for the prominent, but notorious problems of combinatorial auctions and general matroids.


\subsection{A New Framework for Single-Sample Prophet Inequalities}

Our first contribution is a framework for deriving \sspi{}s \emph{directly} via greedy-ordered selection algorithms, as we describe below. This framework can be applied to a range of problems in the \sspi{} setting, two of which we put special focus on: first, we derive \sspi{}s with considerably improved competitive ratios for variants of max-weight matching problems on general/bipartite graphs, extending also to strategic settings.
Second, we employ a richer version of the same design principles to provide the first constant-factor \sspi{} for a combinatorial auction setting, namely combinatorial auctions with \emph{budget-additive} buyers. 
Note that we are able to obtain --- especially the latter result --- only via exploiting 
the high abstraction level and clear structure of probabilistic events within our framework.

\subsubsection{Single-Sample Prophet Inequalities via Greedy-Ordered Selection}

Consider the general problem of selecting online a feasible subset from a number of elements, which arrive in some 
adversarial order. 
We propose a very general and versatile technique for selecting elements in our online algorithms that can be seen as a generalization of several approaches in the literature. At its core, our method is based on two basic and simple insights. As a first ingredient, we use the well-known greedy algorithm, which considers all elements of a set decreasingly by their value and selects an element if and only if the resulting set of chosen elements is still feasible.
The greedy algorithm, although it is offline in the sorting phase, behaves like an online algorithm on the ordered set of elements.  
Our second ingredient is the observation that drawing two values $r_e,\, s_e$ from the distribution of an element $e$ is equivalent to drawing $v_{e,1},\, v_{e,2}$ independently from the same distribution and only subsequently flipping a fair coin to decide which to use as the realization of $r_e$, and which as the sample $s_e$. On a high level, our algorithms work as follows:
\begin{itemize}
    \item{\textit{Greedy Phase.}} First, run a greedy procedure on the sampled values of all elements, and remember for every element $e$ the threshold value $\tau_e$ for which the following holds: element $e$ would have been selected by greedy if we added it to the set of samples with a value $r_e>\tau_e$.
    \item{\textit{Online Phase.}} Second, upon arrival of an element in the adversarial order of \sspi{}, observe its reward $r_e$. If $r_e>\tau_e$, and $e$ together with all the selected elements so far are still feasible, collect $e$.
\end{itemize}

The reason that variations of this simple algorithmic idea yield constant-factor approximations is a combination of our two observations from above. First, the \emph{online} character of greedy after sorting all elements ensures that whether any element is chosen only depends on the \emph{larger} values considered before, but is independent from any later steps, which we heavily rely on in our proofs.
Second, the described process of first drawing \emph{two} values from each element's distribution and only later deciding which is the actual reward relates the quality of the greedy solution on the samples to the online selection process on the actual rewards, and ensures the described thresholds are feasible.

As mentioned before, this type of idea is employed somewhat differently in many existing works. For example, and most closely related to our work, the optimal \sspi{} of \citet{RubinsteinWW20} for the single-choice setting can be interpreted as an implementation of the same paradigm. Also closely related is work by \citet{KP09}, who use greedy pricing combined with the random process of assigning values to the statistics or selection phase of a secretary algorithm (instead of sample or reward value in a \sspi{}) for bipartite weighted matching, and \citet{ma2016simulated}, who apply the greedy idea to submodular secretary problems.

We provide a general and abstract view towards methods of the above type that aids the proofs of our approximation results by making the two key insights presented above very explicit.
First, we use the online-offline-properties of greedy to state an \emph{equivalent offline version} to each of our \sspi{}s, enabling us to much more easily keep track of probabilistic events and the computed thresholds.
Second, we use the view of deciding only at the very last moment which of the two values drawn for an element is the sample or the reward to define (and bound the size of) a set of \emph{safe} elements which even the adversarial arrival order cannot prevent us from picking. The above abstraction is indeed crucial for obtaining our results, since compared to previous work, we face some extra challenges. First, our problems are combinatorially richer than, e.g., the single-choice prophet inequality, requiring careful and dedicated adjustments to the greedy algorithms and thresholds we use for the first algorithm phase.
Second, unlike applications in the secretary paradigm, the arrival orders in our model are fully adversarial, which poses an additional obstacle to guaranteeing a \emph{safe} set in the above sense. Finally, the fact that (again opposed to secretary algorithms) we deal with more than one value associated with the same element (sample and reward) breaks the convenient independence of some of the greedy algorithm's decisions. This requires delicate handling of the resulting dependencies.

\subsubsection{New and Improved Single-Sample Prophet Inequalities}

We now present the results we derive within our greedy framework, which allows us to significantly improve on the competitive guarantees of many existing \sspi{}s, and provide the first constant-factor \sspi{}s for some problems previously not covered by this paradigm.

\paragraph{Max-weight matching with general edge arrivals} In \Cref{sec:generalMatching}, we apply the above ideas to the case of maximum weight matching in general (non-bipartite) graphs with edge arrivals. Here, edges $e\in E$ of a graph $\mathcal{G}$ arrive one by one in adversarial order, each associated with a certain nonnegative weight $w(e)$. The algorithms aim at selecting a maximum-weight matching $M\subseteq E$, such that each vertex has at most one incident edge in $M$. Before our work, the state-of-the-art for this problem was a $512$-competitive algorithm derived via \citet{AzarKW14} and the $256$-competitive \oos{} in \citet{FSZ16} for bipartite graphs, by choosing a random bipartition. We improve this to a $16$-competitive \sspi{} derived through our framework. The proof of this result is exemplary of the key insights that drive our analyses.

\paragraph{Max-weight bipartite matching} In \Cref{sec:bipartite}, we focus on the case of bipartite matching with vertex arrival where the items are available offline and the buyers arrive in adversarial order. We provide an $8$-competitve \sspi{}, again improving considerably over the $256$-competitive \sspi{}.\footnote{We observe however that a \sspi{} with better competitive ratio for this problem can also be achieved by adjusting the analysis in \citet{KP09} to become order-oblivious, which yields a $13.5$ competitive policy.} The improved factor is obtained by exploiting the fact that due to vertex arrival, all edges of the same buyer become available at once. We complement this result with a truthful mechanism, which is based on the idea of tightening the threshold requirements for each edge in a way that agents cannot hurt the approximation too much by picking the \emph{wrong} elements on arrival. In particular, we give a $16$-approximate mechanism, which to the best of our knowledge is the first single-sample result for this setting.
Finally, we show an $8$-competitive \sspi{} for the closely related (and somewhat more restricted) problem of selecting an independent set of elements in a transversal matroid, improving on the previous $16$-approximation by \citet{AzarKW14}.

\paragraph{Budget-additive combinatorial auctions}
In \Cref{sec:budgetadditive} we consider a more general problem, where buyers can be allocated more than one item, and each buyer $b$'s valuation is of the form $v_b(S)=\min\{\sum_{i\in S}v_b(i),\, C_b\}$, for $S$ a set of items and $C_b$ a publicly known, buyer-specific budget.
Here, the buyer's valuation for being assigned additional items changes with the assignment of any previous ones. This dependence makes the design and, to an even greater extent, the analysis of our algorithms considerably more involved, but can still be resolved without incurring too much loss.
As a result, we are able to show a $24$-competitive \sspi{} for this model. This, to the best of our knowledge, is the first \sspi{} for a combinatorial auction setting beyond the simple case of additive valuations, and raises hope that the scope of constant-factor \sspi{}s reaches even further towards that of full-information prophet inequalities. 

\paragraph{Matroids} 
In \Cref{sec:reduction}, we further demonstrate the potential of directly designing \sspi{}s by considering matroids which satisfy a natural partition property. As we observe, many \oos{} policies for matroids rely on decomposing the given matroid into several {\em parallel} (in terms of feasibility) rank-$1$ instances, and then run an \oos~policy (in parallel) on each of these instances. Rather than using the $4$-competitive single-choice \oos{} policy of \citet{AzarKW14}, the same partitioning allows us to leverage the results of \citet{RubinsteinWW20}, which give a $2$-competitive policy for the \sspi~problem on rank-1 matroids. This allows us to improve the competitive guarantees for almost all the matroids considered in \citet{AzarKW14} by a factor of $2$.

\paragraph{Summary of our results and additional properties}

In \Cref{table:main}, we outline our improved \sspi{}s and compare them to previous work. As we have already mentioned, almost all existing \sspi{}s follow from the reduction to \oos{}, due to \cite{AzarKW14}, in combination with existing \oos{} policies (see \citep{DP08} for transversal matroid, \citep{KP09} for graphic matroid and bipartite matching with vertex arrivals, and \citep{Soto11} for co-graphic, low density, and column $k$-sparse linear matroids). 

We remark that (with the exception of budget-additive combinatorial auctions) all of our algorithms are \emph{ordinal} (also known as {\em comparison-based}), in the sense that they do not require accurate knowledge of the values of the random variables, but only the ability to compare any two of them. Further, we emphasize that our competitive guarantees hold against an {\em almighty fully adaptive adversary} -- that is, an adversary who is a priori aware of \emph{all} sample/reward realizations and can decide on the arrival order of the elements in an adaptive adversarial manner. Finally, we highlight that our improved \sspi{}s can be used to provide truthful mechanisms with improved approximation guarantees compared to \citet{AzarKW14} for welfare and revenue maximization.\footnote{Note that, for the case of revenue maximization, and since our algorithms make use of all the available samples, an extra sample is needed to be used as {\em ``lazy reserves''} (see \citet{AzarKW14} for a definition and more details).}

\begin{table}
  \centering
  \begin{tabular}{ccc}
    \toprule
    Combinatorial set     & Previous best     & Our results  
    \\
    \midrule
     General matching (edge arrivals) & $512$ & $16$ 
     \\ 
     Budget-additive combinatorial auction & N/A & $24$ 
     \\
     Bipartite matching (edge arrivals) & $256$ & $16$ 
     \\ 
      & $6.75$ (degree-$d$) & 16 (any degree) 
      \\ 
      & $\O(d^2)$-samples & $1$ sample  
      \\ 
     Bipartite matching (vertex arrivals) & $13.5$ & $8$ 
     \\ 
     Transversal matroid & $16$ & $8$  
     \\ 
     Graphic matroid & $8$ & $4$ 
     \\ 
     Co-graphic matroid & $12$ & $6$ 
     \\
     Low density matroid & $4\gamma(\M)$\footnotemark  & $2\gamma(\M)$  
     \\
     Column $k$-sparse linear matroid & {$4k$} & $2k$  
     \\
    \bottomrule
  \end{tabular}
  \vspace{4pt}
  \caption{Summary of main results}
  \label{table:main}
\end{table}
\footnotetext{The {\em density} of a matroid $\M(E, \I)$ is defined as $\gamma(\M) = \max_{S \subseteq E} \frac{|S|}{r(S)}$, where $r$ is the rank function.}

\subsection{Limitations of Single-Sample Prophet Inequalities}


We have seen that for a wide range of problems, \sspi{}s provide guarantees close to those of traditional prophet inequalities and \emph{direct} \sspi{}s as in our results improve consistently on those derived via the reduction to \oos{}s. 
This positions the \sspi{} paradigm in between full-information prophets and prior-free \oos{}s. 
In consequence, \emph{direct} \sspi{}s raise hope for problems that offer good prophet inequalities, but remain resistant to constant-factor \oos{}s. Here, while we are making some progress with respect to combinatorial auction settings by providing a constant-factor \sspi{} for the budget-additive case, our results on matroids -- although improving in terms of competitive ratio -- do not extend to settings beyond those covered by constant-factor \oos{}s. 
This limitation to our progress is, as we observe, not a coincidence. In fact, we analyze more generally the power of \sspi{}s versus that of \oos{}s, and show the following partial reverse to the reduction of \citet{AzarKW14}: 

\vspace{1em}
\textsc{Theorem}.  (\textsc{informal}) {\em 
For a large and intuitive subclass of \sspi{}s, which we call \emph{pointwise} or \psspi{}s, the existence of a constant-factor \psspi{} policy implies that of an according \oos{} policy.
}
\vspace{1em}

This shows that with our, and all other existing \sspi{} methods, there is no hope to achieve constant-factor \sspi{}s for problems that do not allow for constant-factor \oos{}s. For example, if any of the existing \sspi{} strategies yielded a constant approximation to the general matroid prophet inequality problem, this would also solve the famous matroid secretary conjecture of \citet{BIK07}\footnote{
We further remark that
 our reduction for \psspi{}, together with our \psspi{} for budget-additive combinatorial auction, also implies a constant-factor \oos{} for this problem.}.

\subsection{Further Related Work}
\label{sec:relatedwork}

A great deal of attention has been given to multi-choice prophet inequalities under combinatorial constraints such as uniform matroid constraints \citep{hajiaghayi2007automated,Alaei11} and general matroid constraints \citep{ChawlaHMS10,KleinbergW12,FSZ16}.
A number of works have obtained prophet inequalities for matchings and combinatorial auctions, including 
\citep{AHL12,FSZ14,ehsani2018prophet,GravinW19,DuttingFKL20,DuttingKL20,EFGT20,CorreaC23}. In particular, \citet{GravinW19} provide prophet inequalities for weighted bipartite matching environments under edge arrivals, and show a lower bound of $2.25$ on the competitive ratio of any online policy for this setting, proving that this problem is strictly harder than the matroid prophet inequality. Quite recently, \citet{EFGT20} provided a near-optimal prophet inequality for general weighted {\em(non-bipartite) matching} under edge arrivals. Beyond these settings, recent work has considered the prophet inequality problem under arbitrary packing constraints \citep{R16,RS17}.

Regarding the prophet inequality problem in the limited information regime, in addition to their meta-result connecting \sspi{}s with \oos algorithms, \citet{AzarKW14} provide a threshold-based $(1-\O(\nicefrac{1}{\sqrt{k}}))$-competitive algorithm for $k$-uniform matroids.
Furthermore, \citet{RubinsteinWW20} develops a $(0.745 - \O(\epsilon))$-competitive algorithm for the single-choice IID case, using $\mathcal{O}(\nicefrac{n}{\epsilon^6})$ samples from the distribution, improving on the results of \citet{correa2019prophet}. 
\citet{Guo0T021} further improve this bound to $O(\nicefrac{n}{\epsilon^2})$ samples, while \citet{CorreaCES24} show that $O(\nicefrac{n}{\epsilon})$ samples suffice. In very recent work, \citet{CristiZ24} show that a constant number of samples from each distribution are also sufficient to get $\epsilon$-close to the optimal ratios for the random-order and the free-order prophet inequality problem.
Prophet inequalities that require a polynomial number of samples can also be obtained via the balanced prices framework \citep{FeldmanGL15,DuttingFKL20}.

A main motivation for studying the prophet inequality problem comes from connections to mechanism design. While optimal single-parameter mechanisms for both welfare  \citep{v61,g73,c97} and revenue \citep{M81} have been well-understood for decades, there has been an extended study of simple and practical mechanisms that approximate these objectives. Prophet inequalities are known to be a powerful tool for designing simple posted-price mechanisms.
\citet{AzarKW14} show how to apply the results of \citet{ChawlaHMS10,amdw13,dry10} to obtain mechanisms in settings where the mechanism designer only has access to a single sample (or a constant number of samples) from the distribution of the agents' values.
A recent work of \citet{Duetting20} also investigates single-sample mechanism design for two-sided markets. The authors prove matching upper and lower bounds on the best approximation that can be obtained with one single sample for subadditive buyers and additive sellers. Moreover, computationally efficient black-box reductions that turn any one-sided mechanism into a two-sided mechanism with a small loss in the approximation, while using only one single sample from each seller, are provided.

In concurrent and independent work, a greedy technique similar to our framework was implemented for the problem of online weighted matching by \citet{KaplanNR22}. Their main model, \emph{adversarial-order model with a sample}, differs from ours in that it is closer to secretary algorithms than \sspi{}: instead of having two values for each element, reward and sample, they assume the bipartite graph to consist of a \emph{historic} part that is available to the algorithm, and the actual problem instance. Using an intermediate step, the so-called two-faced model, they derive constant factor approximation for \sspi \ for bipartite and general matching. In particular, they show a $13.5$-approximation for the general graph, edge arrival model, and a $3 + 2\sqrt 2 \approx 5.83 $-approximation for the bipartite, vertex arrival model. 
 In the most recent version of their paper they improve the $13.5$ bound to $11.66$ \cite{KaplanNR22}.

\section{Edge Arrival in General Graphs}
\label{sec:generalMatching}
In this section, we consider online matching on a general graph $\mathcal{G}=(V,E)$ with edge arrivals, where the weight of each edge $e \in E$ is drawn independently from a distribution $D_e$. We recall that for the offline problem where the edge weights are known a priori, the greedy algorithm that collects edges in a non-increasing order of weight while maintaining feasibility is a $2$-approximation. 
Our algorithm outlines very well the idea we described above: to bridge the gap between the fixed order employed by greedy and the adversarial one of prophet inequalities (where greedily adding edges would be arbitrarily bad), we utilize the independent single samples $S=\{s_e \mid e\in E\}$ that we have from all the edge distributions.

Note that throughout we consider \emph{weighted sets} to denote our edge sets and allocations, in the sense that they also contain -- for each edge $e$ -- the associated weight $w(e)$. Furthermore, we assume to break ties uniformly at random before running the algorithm, i.e. for the set of all drawn numbers that have value $x$, pick a permutation $\pi$ of those uniformly at random and define a draw of the same value to be larger than another if and only if it comes first in $\pi$.

The algorithm first computes an {\em offline} greedy matching $\Ms$
on the graph with edge weights $S$. Then, for each vertex $v$, it interprets the weight of the edge incident to it in the greedy solution as a {\em price} $p_v$ for the adversarial-order online algorithm, with the convention that  $p_v = 0$ if there is no edge in $\Ms$ incident to $v$. In the online phase, whenever an edge $e = \{u,v\}$ arrives with weight $r_e$ and both endpoints are free (i.e., not already matched), it is added to the matching $\M$ if and only if $r_e \ge \max\{p_v, p_u\}$. To facilitate the analysis, together with the solution $\M$, the algorithm builds $\Eprime$, a subset of $E$ containing all the edges such that $r_e \geq \max\{p_v, p_u\}$, i.e., all edges that are price-feasible. Note that $\Eprime$ is a superset of the actual solution $\M$ and may not be a matching. See \Cref{AlgoOn} for a pseudocode.

Although the above algorithm is simple, the analysis is quite challenging.
As pointed out before, the greedy method can be seen as an offline algorithm as well as an online algorithm with the extremely simplifying assumption of a fixed and weight-decreasing arrival order. We combine such a shift of viewpoint with a suitable reformulation of the random processes involved. More precisely, we look at our (online) algorithm as the byproduct of a run of a (fixed-order!) greedy strategy on a random part of a fictitious graph that contains each edge \emph{twice}: once with the true weight $r_e$, and once with the sampled weight $s_e$.

\begin{algorithm2e}[!t]
\DontPrintSemicolon
\caption{Prophet Matching with Edge-Arrivals}
    Set $\Eprime=\emptyset$, $\M= \emptyset$\;
    Compute the greedy matching $\Ms$ on the graph with edge weights according to sample $S$\;
    \For{each $e = \{u,v\} \in \Ms$}{
    Set $p_{u}=p_{v}=s_e$}
    \For{each vertex $u \in V$ unmatched in $\Ms$} {
    Set $p_u=0$}
    \For{each arriving $e = \{u,v\} \in E$}{
        \If{$r_e \geq \max\{p_u,p_v\}$}{
            Add $e$ to $\Eprime$. \tcp{This is just for the analysis.}
            \If{$u,v$ are not matched in $\M$}{
                Add $e$ to $\M$ with weight $r_e$
            }
        }
    }
    \Return $\M$
\label{AlgoOn}
\end{algorithm2e}

\paragraph{Equivalent offline algorithm} To analyze the competitive ratio of \Cref{AlgoOn}, we state an offline, greedy-based algorithm with similar properties that is easier to work with.
This offline procedure considers an equivalent stochastic process generating the edge weights in $S$ and $R$. Instead of first drawing all the samples and then all the realized weights, we consider a first stage in which, for each edge $e$, two realizations are drawn from $D_e$. Then, all drawn realizations ($2 |E|$ in total) are sorted in decreasing order according to their values, breaking ties at random as explained above. 
Finally, the offline algorithm goes through all these values in that order; each value is then associated with $S$ or $R$ by a random toss of an unbiased coin. Consequently, the second (i.e., smaller) weight of the same edge arriving later is associated to the remaining category with probability $1$. In particular, an edge $e$ is marked as $R$-used, when $r_e > s_e$, and $S$-used, otherwise. It is not hard to see that the distributions of $S$ and $R$ are identical for both random processes. The pseudocode for this offline procedure is given in \Cref{AlgoOff}. As already mentioned, it mimics in an offline way the construction of $\Ms$ and $\Eprime$ of \Cref{AlgoOn}, as shown below.

\begin{algorithm2e}[!t]
\DontPrintSemicolon
\caption{Offline Simulation for Matching with Edge-Arrivals}
\label{AlgoOff}
    Set $\Eprime=\emptyset$, $\Ms= \emptyset$, $\M = \emptyset$, and $V^S=V$ \;
    For each $e\in E$, draw from $D_e$ two values $a_{e,1}$ and $a_{e,2}$\;
    Order $A=\{a_{e,1}, a_{e,2}|e\in E\}$ in decreasing fashion \;
    \For{each value $a\in A$ in the above order}{
        \If{$a$ corresponds to edge $e = \{u,v\}$ that has never been observed before}{
        {Flip a fair coin \;
            \eIf{Heads}
            {
                Mark $e$ as \emph{$R$-used} \;
            \If{$u \in V^S$ and $v \in V^S$}{
          Add $e$ to $\Eprime$ with  weight $a$\;}
            }
            {
        Mark $e$ as \emph{$S$-used}\;
        \If{$u \in V^S$ and $v \in V^S$}{
        Add $e$ to $\Ms$ with weight $a$, remove $u$ and $v$ from $V^S$\;
        }
        }        
        }
        }
        \ElseIf{$a$ corresponds to edge $e = \{u,v\}$ which is $R$-used, $u\in V^S$ and $v\in V^S$}{
        Add $e$ to $\Ms$ with weight $a$, remove $u$ and $v$ from $V^S$\;
        }
    }
    \For{each $e$ in the same order as \Cref{AlgoOn}}{
        \If{$e$ in $\Eprime$ and $\{e\} \cup \M$ is a matching }
        {Add $e$ to $\M$ with weight $r_e$}}
    \Return $\M$.
\end{algorithm2e}

\begin{claim}
\label{claim:equivalence}
    The sets $\Eprime$, $\Ms$, and $\M$ are distributed in the same way when calculated by \Cref{AlgoOn} as when calculated by \Cref{AlgoOff}.
\end{claim}
\begin{proof}
To begin, note that the distribution of all the rewards in $R$ and the samples in $S$ is exactly the same in both \Cref{AlgoOn} and \Cref{AlgoOff}. This follows by the fact that, for each edge $e$, the distribution of $r_e$ remains the same if we simply sample $r_e$ from $D_e$, independently, or if we sample two values $a_{e,1}$ and $a_{e,2}$ from $D_e$, independently, and then set $r_e$ equal to one of the two equiprobably. The symmetric argument holds for any sample value in $S$. 

We prove that for {\em any} realization of the $S$ and $R$ values, the corresponding sets of edges of interest are exactly the same in the offline and online setting. This is clearly enough to conclude the proof. 
To avoid complication, we deterministically use the same tie breaking rules in both the online and offline case. 
    
If we restrict the offline algorithm to consider only samples $S$, we recover exactly the procedure to generate a greedy matching with respect to the samples. This is true because what happens to the values in $R$ has no influence on $\Ms$. Therefore, the two versions of $\Ms$ in the different algorithms follow the same distribution. Actually, this is even true step-wise: at any point in time during \Cref{AlgoOff}, $\Ms$ contains the greedy matching on samples {\em restricted} to the edges whose $s_e$ is greater than the value being considered at that specific time. 
    
Now, we show that the two versions of $\Eprime$ coincide. This is enough to conclude the proof, as the matchings $M$ are extracted in the same way from $\Eprime$ (given the same fixed arrival order). To this end, consider an (online) iteration of \Cref{AlgoOn} in which an edge $e=\{u,v\}$ is added to $E'$. In such event, the online version of $\Ms$ contains no edges incident to $u$ or $v$ with larger sample than $r_e$, because this would imply a price higher than $r_e$ in at least one of the two endpoints. Since the $\Ms$ are the same in both algorithms, when \Cref{AlgoOff} considers $e$, it holds that $u,\, v\in V^S$ ($u$ and $v$ are not yet matched in $\Ms$) and \Cref{AlgoOff} will also add $e$ to $E'$. The converse argument is immediate.
\end{proof}
Given the above result, we use in the following $\Eprime$, $\Ms$, and $\M$ mainly to refer to \Cref{AlgoOff}, but the same would apply to \Cref{AlgoOn}.

\paragraph{Correctness and competitive analysis} 

The correctness of \Cref{AlgoOn} is guaranteed by the fact that the algorithm never adds to $\M$ any edge incident to an already matched vertex. Thus, the set of edges collected in $\M$ is a valid matching. For the competitive analysis we need some more work. In particular, we show that \Cref{AlgoOn} is $16$-competitive, namely, for any adversarial arrival order of the edges, the algorithm retains in expectation at least a $\nicefrac1{16}$-fraction of the expected value of optimal offline matching $\opt$. The high-level idea of the proof lies in relating the optimal offline matching and our solution by using $\Ms$ (which yields a $2$-approximation of $\opt$) and a carefully chosen subset of $\Eprime$, the {\em safe} edges. 

\begin{definition}[Safe edges] \label{def:safe}
    We call an edge $e = \{u,v\}$ safe for its endpoint $v$ if the following conditions are met:
    \begin{itemize}
        \item[(i)] $e$ is the only edge in $\Eprime$ incident to $v$.
        \item[(ii)] There is no edge in $\Eprime$ that is incident to $u$ and has smaller weight than $r_e$.
    \end{itemize}
\end{definition}

Let us denote $e_v$ as the edge incident to $v$ in $\Eprime$ of maximal reward (if such an edge exists in $\Eprime$). We are able to show the following convenient property for this edge being safe for $v$:


\begin{lemma}\label{lem:matching:safeprob}
For any vertex $v$ and edge $e$ incident to $v$ for which the probability that $e_v=e$ is non-zero, 
we have that $    \mathbb{P}[e_v \text{ is safe for }v \mid e_v = e] \geq \frac{1}{4}.$ 
\end{lemma}
\begin{proof}
 Recall that an edge is marked as $R$-used if its realized reward is greater than the corresponding sample. In the opposite case, the edge is marked as $S$-used. By construction of \Cref{AlgoOff}, $\Eprime$ contains only $R$-used edges. Now consider the time in which some edge $e_v = e = \{v,u\}$ with value $a$ is added to set $\Eprime$. By definition of $e_v$, this is the first edge added in $\Eprime$ that is incident to vertex $v$. Consider what happens the next time some edge $e'$ of value $a'$, incident to either $u$ or $v$ (or both), and such that both endpoints of $e'$ are in $V^S$, is parsed by \Cref{AlgoOff}. We distinguish among the following cases:

\paragraph{Case 1: $e' = \{v,u'\}$, where $u' \neq u$} In this scenario, we see that $e' = \{v,u'\}$, that is, $e'$ is incident to vertex $v$ but not $u$. We now show that, with probability at least $\nicefrac{1}{2}$, the first condition of \Cref{def:safe} is satisfied, by considering two subcases:
\begin{itemize}
    \item[\bf(a)] In the case where $e'$ has been observed before (i.e., $a'$ is its second realization), then $e'$ is $R$-used, and thus $a'$ is a sample value. Indeed, if $e'$ were $S$-used, then at least one of the its endpoints would no longer be in $V^S$, a fact which would contradict the choice of $e'$. Therefore, since $e'$ is $R$-used and both its endpoints are in $V^S$ by the time $a'$ is parsed, vertex $v$ (the only common vertex between $e$ and $e'$) is now removed from $V^S$ and therefore no other edge incident to $v$ can be added to $E'$. Thus, $e_v$ must be the only edge in $\Eprime$ incident to $u$, as desired.
    \item[\bf(b)] In the case where $e'$ has not been observed before by \Cref{AlgoOff}, it is easy to see that, with probability $\nicefrac{1}{2}$, $e'$ is $S$-used. Therefore, by the same argument as above, $e_v$ is the only edge in $E'$ incident to $v$. Once the first condition of \Cref{def:safe} is satisfied, let $e'' = \{u,u''\}$ be the next edge of value $a'' < a' < a$ that is parsed by \Cref{AlgoBipOff}, such that $u,u''\in V^S$. By repeating the exact same arguments as in subcase (a), we can see that the second condition of \Cref{def:safe} will additionally be satisfied with probability at least $\nicefrac{1}{2}.$ Hence, we can see that the probability that $e_v$ is safe for $v$ is at least $\nicefrac{1}{4}$. 
\end{itemize}

\paragraph{Case 2: $e' = \{v',u\}$, where $v' \neq v$}
Using exactly the same arguments, the above analysis can be replicated, even if the first parsed edge after the time where $e_v$ is added to $\Eprime$ is incident only to vertex $u$, but not $v$.

\paragraph{Case 3: $e' = \{v,u\}$} In this case, $e'$ corresponds to an edge parallel to $e_v$ or to the edge $e_v$ itself. If $e'$ is a parallel edge to $e_v$, then, by an analysis similar to the first case, both conditions are satisfied with probability at least $\nicefrac{1}{2}$. On the other hand, if $e'$ coincides with $e_v$ (which is by definition $R$-used), then it can be verified that $e_v$ is already safe for $v$ (with probability $1$).
\end{proof}
Recall that $e_v$ is defined to be the edge of maximal reward among the edges in $\Eprime$ incident to vertex $v$ (if such an edge exists). We denote by $r_{e_v}$ the reward for this edge and set $r_{e_v} = 0$ if this edge does not exist. We relate the weight of $M$ and that of the sum of the $r_{e_v}$ in the following two lemmas.

\begin{lemma}\label{lem:matching:safectononsafe}
For any realization of the edge weights in $A$, we have
\begin{align*}
    \EXP{\sum_{v\in V} r_{e_v} \1{e_v\text{ is safe for }v}} 
    &\geq  \frac{1}{4}  \cdot \E{\sum_{v \in V} r_{e_v}}.
\end{align*}
\end{lemma}
\begin{proof}
Let us fix any realization of $A$ and assume w.l.o.g. that for each edge $e \in E$ with realizations $a_{e,1}$ and $a_{e,2}$, it holds $a_{e,1} > a_{e,2}$. 
We show that that for any vertex $v \in V$, it holds $\EXP{r_{e_v} \1{e_v\text{ is safe for }v}} \geq \frac{1}{4} \cdot \EXP{r_{e_v}}$. By linearity of expectations, this is enough to prove the Lemma.

It is not difficult to verify that for each edge $e \in E$ incident to vertex $v$,  
if $e$ is safe for $v$, then $r_e = a_{e,1}$, the larger weight realization.
Indeed, assuming that $r_e = a_{e,2}$, the value $a_{e,2}$ is parsed after $a_{e,1}$ by \Cref{AlgoOff} and thus it cannot be the case that $e \in \Eprime$. Given this fact, we have the first passage of the following display:
\begin{align*}
    \EXP{r_{e_v} \1{e_v\text{ is safe for }v}}  &= \sum_{e \in E} a_{e,1} \cdot \PRO{e_v\text{ is safe for }v\text{ and }e_v=e}\\
    &= \sum_{e \in E} a_{e,1} \cdot \PRO{e_v\text{ is safe for }v \mid e_v=e} \cdot \PRO{e_v=e}\\
    &\ge \frac 14 \sum_{e \in E} a_{e,1} \PRO{e_v=e}\\
    &=\frac{1}{4} \cdot \EXP{r_{e_v}},
\end{align*}
where the inequality follows from  \Cref{lem:matching:safeprob}.
\end{proof}

\begin{lemma} \label{lem:matching:safecollect}
In any run of \Cref{AlgoOff}, for the expected reward collected in $M$, we have
\[
\mathbb{E}\left[ w(M) \right]
    \geq 
    \frac{1}{2} \cdot \mathbb{E}\left[\sum_{v\in V} r_{e_v} \1{\text{$e_v$ is safe for $v$}}\right].
\]
\end{lemma}
\begin{proof}
We first partition the set $M \subseteq \Eprime$ into two sets $M^B$ and $M'$. Informally, $M^B$ is the subset of edges $e = \{u,v\}$ in $M$ that are safe for both endpoints $u$ and $v$, whereas $M'$ contains the rest of the edges in $M$. Formally: 
\[
M^B = \{e = \{u,v\} \in M~\mid~ e \text{ is safe for }u \text{ and for } v\} ~~\text{   and   }~~~ M' = M \setminus M^B.
\]
Similarly, we denote by $V^B$ the subset of vertices that are adjacent to some edge in $M^B$, and by $V' = V \setminus V^B$ the rest of the vertices. 

By definition of safe edges, each edge $e = \{u,v\}$ in $M^B$ is the only edge that can be collected by both endpoints $u$ and $v$ in the second phase of \Cref{AlgoOff}. Further, since the weight of such an edge is contributes to both its endpoints, it follows directly that 
\[
w(M^B) 
=
\frac{1}{2} \sum_{v \in V^B} r_{e_v} \1{e_v\text{ is safe for } v}.
\]
Turning now our attention to the set $M'$, consider any vertex $v \in V'$ that is incident to some edge $e_v = \{v,u\}$ that is safe for $v$
but not for $u$.
By definition of the safe edges, by the time this edge arrives in the second phase of \Cref{AlgoOff}, it is either collected by the algorithm, or another edge $e' = \{u,u'\}$, incident to $u$, has already been collected with reward $r_{e'}$. However, since $e_v$ is safe for $v$, it has to be that $r_{e'} \ge r_{e_v}$, since no edge in $\Eprime$ incident to $u$ has smaller reward than $e_v$. Note that $e' = \{u,u'\}$ cannot be safe for $u$, given that in the above scenario 
$\Eprime$ contains at least \emph{two} edges incident to $u$, which violates \Cref{def:safe}.
Further, $e'$ cannot be safe for $u'$ either, since $\Eprime$ contains $e_v$, which is incident to $u$ and has reward \emph{smaller} than $r_{e'}$, violating the second condition of \Cref{def:safe}.
Finally, adding any edge $e$ to $M'$ can make at most two safe edges infeasible: either the added edge $e$ was safe itself, or after adding $e$, we can no longer include any safe edges incident to its endpoints, each of which has at most one safe, incident edge by definition. We can conclude that
\[
w(M') \geq \frac{1}{2} \sum_{v \in V'} r_{e_v} \1{e_v\text{ is safe for } v}.
\]
By combining the above expressions, and since $M^B$ and $M'$ (resp., $V^B$ and $V'$) partition $M$ (resp., $V$), we get
\[
    w(M) = w(M^B) + w(M') \geq \frac{1}{2} \sum_{v \in V} r_{e_v} \1{e_v\text{ is safe for } v}.  
\]
The proof of the lemma follows simply by taking the expectation. 
\end{proof}
Now that the preparatory lemmas have been proved, we can present our main result of the Section.

\begin{theorem}
\label[theorem]{thm:edgeGen}
    For the problem of finding a maximum-weight matching in a general graph $\mathcal{G}$, in the online edge arrival model, \Cref{AlgoOn} is $16$-competitive in expectation, i.e. $ 16 \cdot \E{w(\M)} \ge \E{\opt},$ where $\opt$ is the weight of an optimal matching in $\mathcal{G}$.
\end{theorem}
\begin{proof}
We define for any $v \in V$ and any any \emph{weighted} edge set $\hat{E}\subseteq E$ the quantity $\val{\hat{E}}{v}$, that corresponds to the largest weight of an edge incident to vertex $v$ in edge set $\hat{E}$, or $0$ if there is no such edge. Fix any vertex $v \in V$ and consider the first time in the run of \Cref{AlgoOff}, where an edge $e = \{v,u\}$ of value $a$ arrives which is incident to $v$ and, at that time, $u \in V^S$. Note that in the case where $u \notin V^S$, the edge cannot be added to neither $\Eprime$ nor $\M_s$. By construction of the offline algorithm and since $v \in V^S$, with probability $\nicefrac{1}{2}$ (i.e., the probability of the fair coin-flip), edge $e$ is marked as $R$-used and it is added in $E'$. 
 At this point during the run of \Cref{AlgoOff}, no matter the outcome of the fair coin flip, at most one edge incident to $v$ can be added to $\Ms$ -- either $e$ with weight $a$, or some other edge $e'$ of weight \emph{smaller} than $a$.
Thus, for any $v \in V$, we have that $\mathbb{E}\left[\valEprime{v}\right]  
\geq \frac{1}{2} \cdot \mathbb{E}\left[\valMs{v}\right]$, which, by linearity of expectation implies that

\[
    \EXP{\sum_{v\in V} r_{e_v}}
    = \mathbb{E}\left[\sum_{v\in V} \valEprime{v}\right]
    \geq \frac{1}{2} \cdot \mathbb{E}\left[ \sum_{v\in V}\valMs{v}\right]
    = \mathbb{E}\left[w(M_S) \right].
\]
By combining \Cref{lem:matching:safectononsafe,lem:matching:safecollect}, for the expected reward collected in $\M$

\begin{align*}
    \mathbb{E}\left[ w(M) \right]
    & \geq 
    \frac 12 \cdot \E{\sum_{v\in V} r_{e_v} \1{e_v\text{ is safe for }v}} \\
    &\ge \frac 18 \cdot \E{\sum_{v \in V} r_{e_v}} \ge \frac{1}{8} \cdot \E{w(\Ms)}
    \ge \frac{1}{16} \cdot \E{\opt},      
\end{align*}
where in the last inequality we use the fact that $w(\Ms)$ has the same distribution as the value of the greedy solution with respect to the rewards, which in turn is a $2$-approximation of the optimal matching.
\end{proof}

\section{Vertex Arrival in Bipartite Graphs} \label{sec:bipartite}

We now consider the case of \sspi{} matching on bipartite graphs with vertex arrivals. More specifically, we consider a bipartite graph $\mathcal{G} = (B \cup I, E)$, where $B$ is the set of left vertices, which we refer to as ``buyers'', and $I$ is the set of right vertices, which we refer to as ``items''. Every edge $e = \{b, i\}$ between buyer $b \in B$ and item $i \in I$ is associated with an independent distribution $D_e$. In the offline phase, the gambler receives a single sample $S$ from the product distribution of edge weights. In the online phase, buyers arrive sequentially and in adversarial order, while the gambler observes simultaneously the rewards of all edges incident to the arriving buyer. At each time, the gambler decides irrevocably whether to match the buyer to some unmatched item, or skip. The objective is to maximize the total expected reward collected by the buyers, against that of a prophet who can compute a max-weight matching w.r.t. the rewards $R$.

In this section, we first provide a $\O(1)$-competitive SSPI for the above model. Then, we show how our algorithm can be translated into a truthful mechanism where buyers are incentivized to report the true reward for each item. Finally, we show how our analysis extends to the case where the weights of all the edges incident to the same buyer are identical (i.e., buyer-dependent), which allows us to provide an \sspi{} of improved competitive guarantee for the case of transversal matroids.

\subsection{Main algorithm and analysis} \label{sec:bipartitemain}
    
\Cref{AlgoBipOn} for the bipartite graph model operates in a similar manner as that for general graphs of \Cref{sec:generalMatching}. As before, the sample $S$ is used to calculate prices on both sides: items (fixed side vertices) have prices to \emph{protect} them from being sold too cheaply, while buyers (incoming vertices) have prices they need to beat in order to participate in the market, in addition to beating the item prices. We denote by $M_S$ the greedy matching on the $S$ graph, while by $\Eplus$ the set of edges of maximum reward (greater than their thresholds) for each arriving buyer $b\in B$. Let $M$ be the matching returned by the algorithm on the edges of $\Eplus$, where conflicts are solved in favor of the first arriving edge. Note that $\M$ is computed in an online, adversarial ordering.

\begin{algorithm2e}[!t]
\DontPrintSemicolon
    Set $\Eplus=\emptyset$, $M=\emptyset$\;
    Compute a greedy matching $\Ms$ on the graph with edge weights $S=\{s_e \mid e\in E\}$\;
    \For{each $e = \{b, i\} \in \Ms$}{
        Set $p_{i}=p_{b}=s_e$
        }
    \For{each vertex $k$ not matched in $\Ms$}{
        Set $p_{k}=0$.
        }  
    \For{each arriving buyer $b \in B$}{
        Let $\hat e = \{b, i^*\} = \argmax \{ r_e \mid \ e = \{b, i\}\text{ and } r_e \geq \max\{p_{b},p_i\}\}$\;
        Add $\hat e$ to $\Eplus$ with weight $r_{\hat{e}}$     \tcp{Only used for the analysis.}
        \If{$i^*$ is not matched in $\M$}{
            Add $\{b, i^*\}$ to $\M$ with weight $r_{\hat e}$
        }
    }
    \Return $\M$
\caption{Bipartite Prophet Matching with Vertex-Arrivals}
\label{AlgoBipOn}
\end{algorithm2e}

\paragraph{Equivalent offline algorithm} In order to relate the weight of $\Eplus$ and $\M$ with the prophet's expected reward, similarly to the previous section, we consider an offline version of the algorithm (\Cref{AlgoBipOff}), which exhibits the same distribution over the relevant sets, and can be analyzed more easily. Indeed, this offline algorithm interleaves the building of a greedy matching $\Ms$ with that of the allocation $\Eplus$. While all the values are drawn in advance, they are assigned to $S$ and $R$ only upon arrival. We establish the following equivalence:

\begin{claim}
\label{claim:equivalenceBipartite}
    The sets $\Eplus$, $\Ms$ and $\M$ are distributed the same way when computed by \Cref{AlgoBipOn} as when computed by \Cref{AlgoBipOff}.
\end{claim}
\begin{proof}
The sets $R$ and $S$ in the offline and online setting follow the same distribution as argued in \Cref{claim:equivalence}, and $\Ms$ is still the greedy matching with respect to the samples $S$, in both \Cref{AlgoBipOn} and \Cref{AlgoBipOff}. We know that $\M$ is extracted in the same way by $\Eplus$ by the two algorithms, so we just need to argue that for each fixed realization of the edge weights, the online and offline versions of the edge sets $\Eplus$ coincide.
    
    In \Cref{AlgoBipOn}, $\Eplus$ contains \emph{at most} one edge $e=\{b,i\}$ for each buyer $b\in B$. This is the edge of \emph{largest} reward that is above both prices $p_b$ and $p_i$, if such an edge exists.
    Now, in \Cref{AlgoBipOff}, observe that, once an edge $e=\{b,i\}$ is added to $\Eplus$, the buyer $b$ is removed from the set $B^R$. Therefore, in this algorithm also, at most one edge can be added to $\Eplus$ for each buyer. Now, observe that if the edge is added to $\Eplus$, then it must be over the price of $b$ and $i$.
    Indeed, this follows because, in order for $e$ to be added to $\Eplus$, it must be the case that $b\in B^R$ (and, by construction, also in $B^S$) and $i\in I^S$,
    which implies that $b$ and $i$ have not yet been matched in the greedy sample solution, and thus, because of the greedy traversal order, its reward must be above both prices.
    Finally, we observe that, again due to the greedy traversal order, if $e$ is added to $\Eplus$, then it must be the \emph{largest}-reward edge which is above the prices $p_b$ and $p_i$.
    Hence $\Eplus$ follows the same distribution in both algorithms.
\end{proof}

\paragraph{Correctness and competitive analysis}
The correctness of \Cref{AlgoBipOn} is guaranteed since only edges in $\Eplus$ can be added in $\M$, $\Eplus$ contains at most one edge for each buyer, and the algorithm never assigns an item to more than one buyers. Hence, the set of collected edges in $M$ is a valid matching. Furthermore, we show now that \Cref{AlgoBipOn} is $8$-competitive. As in \Cref{sec:generalMatching}, the main idea of the proof is to relate the weight of the optimal offline matching, $M_S$, with the matching collected by our algorithm, $M$, using a carefully chosen subset of $\Eplus$, the \emph{safe} edges.

\begin{definition}[Safe edges for bipartite graphs]\label{def:safeBipartite}
    We call an edge $e=\{b,i\}\in E$ \emph{safe for buyer $b$} if the following conditions are true:
    \begin{itemize}
        \item[(i)] $e$ is the only edge in $\Eplus$ incident to buyer $b$.
        \item[(ii)] No edge in $\Eplus$ incident to item $i$ has smaller weight than $r_e$.
    \end{itemize}
\end{definition}

As in \Cref{sec:generalMatching}, we denote $e_b \in \Eplus$ the edge incident to buyer $b$, among those in $\Eplus$, if such an edge exists, in a run of \Cref{AlgoBipOff}. Additionally, we fix a realization of $A$, $b\in B$, and $e \in E$ for which the probability that $e_b = e$ is non-zero.

\begin{algorithm2e}[!t]
\DontPrintSemicolon
    Set $\Eplus=\emptyset$, $\Ms= \emptyset$, $\M=\emptyset$, $B^S=B^R=B$, and $I^S=I$\;
    For each $e\in E$, draw from $D_e$ two values $a_{e,1}$ and $a_{e,2}$\;
    Order $A=\{a_{e,1}, a_{e,2} \mid e\in E\}$ in a decreasing fashion \;
    \For{each value $a\in A$ in the above order}{
    \uIf{$a$ corresponds to edge $e = \{b,i\}$ that has never been observed before}{
            Flip a fair coin\;
            \uIf{Heads}{
                Mark $e$ as $R$-used\;
                \uIf{$b\in B^R$ and $i\in I^S$}{
                    Add $e$ to $\Eplus$ with weight $a$, remove $b$ from $B^R$\;
                }
            }
            \uElse{
                Mark $e$ as $S$-used\;
                \uIf{$b\in B^S$ and $i\in I^S$}{
                    Add $e$ to $\Ms$ with weight $a$, remove $b$ from $B^S$ and $B^R$, and $i$ from $I^S$\;
                }
            }
        }
        \uElseIf{$a$ corresponds to an edge $e=\{b,i\}$ which is $R$-used, $b\in B^S$ and $i\in I^S$}{
            Add $e$ to $\Ms$ with weight $a$, remove $b$ from $B^S$ and $i$ from $I^S$\; 
        }
    }
    \For{each $e$ in the same order as \Cref{AlgoBipOn}}{
        \If{$e \in \Eplus$ and $\{e\} \cup \M$ is a matching }{ 
            Add $e$ to $\M$ with weight $r_e$\;
        }
    }
    \Return $\M$.
\caption{Offline Simulation for Bipartite Matching with Vertex-Arrivals}
\label{AlgoBipOff}
\end{algorithm2e}

\begin{lemma}\label{lem:lowerBoundProbBipartite}
    For any buyer $b \in B$ and edge $e$ incident to $b$ for which the probability that $e_b = e$ is non-zero, we have that $ \PRO{\text{$e_b$ is safe for $b$} \mid e_b = e} \geq \nicefrac{1}{2}$.
\end{lemma}
\begin{proof}
As in the proof of \Cref{lem:matching:safeprob}, we note that an edge is marked $R$-used if its realized reward is greater than the corresponding sample, and otherwise is marked as $S$-used. By construction, \Cref{AlgoBipOff} only adds $R$-used edges to $\Eplus$. Let us consider the time in the greedy ordering when $e_b = e = \{b,i\}$ is added to set $\Eplus$. By definition of $e_b$, this is the \emph{first} time an edge incident to buyer $b$ is added to $\Eplus.$ Additionally, by construction of the algorithm, this must be the \emph{only} edge corresponding to buyer $b$ which is added to $\Eplus$. Thus, the first condition of \Cref{def:safeBipartite} is satisfied by construction.

Now, let us consider the next time in the greedy ordering when an edge $e'=\{b',i\}$ with weight $a'$ is parsed by the algorithm, $b'\in B^S$, and $i\in I^S$ (indeed, if $b'\not\in B^S$ or $i\not\in I^S$, then $e'$ could not be added to either $\Eplus$ or $\Ms$). If this is the second time that $e'$ has been processed by the algorithm, then it must be the case that $e'$ is $R$-used (since, otherwise, we get a contradiction to the choice of $e_b$), and hence, $a'$ is a sample value. At this time, $b'$ is removed from $B^S$ and $i$ from $I^S$, so no other edges incident to item $i$ can be added to $\Eplus$ after this point. Thus, in this case, both conditions of \Cref{def:safeBipartite} are satisfied.
Now, in the case when $e'$ is being processed for the \emph{first} time, we have that, with probability $\nicefrac{1}{2}$, $e'$ is $S$-used, thus guaranteeing that, by similar arguments as above, both conditions of \Cref{def:safeBipartite} are satisfied.
\end{proof}
In the following two lemmas, we relate the expected weight of the matching $M$ with a constant fraction of the expected sum of the $r_{e_b}.$

\begin{lemma}\label{lem:safeRewardBipartite}
    For any fixed realization of edge weights $A$, we have 
    \begin{align*}
        \EXP{\sum_{b\in B} r_{e_b} \1{\text{$e_b$ is safe for $b$}}}
        \geq
        \frac{1}{2} \cdot \EXP{\sum_{b\in B} r_{e_b}}.
    \end{align*}
\end{lemma}
\begin{proof}
    The proof follows by the exact same arguments presented in \Cref{lem:matching:safectononsafe}, replacing the application of \Cref{lem:matching:safeprob} with \Cref{lem:lowerBoundProbBipartite}.
\end{proof}

\begin{lemma}\label{lem:algRewardBipartite}
    In any run of \Cref{AlgoBipOff}, for the expected reward collected in $\M$, we have
    \begin{align*}
        \EXP{w(\M)} 
        \geq
        \EXP{\sum_{b\in B} r_{e_b} \1{\text{$e_b$ is safe for $b$}}}.
    \end{align*}
\end{lemma}
\begin{proof}
    Consider any fixed buyer $b$, and suppose that the edge $e_b = \{b,i\}$ is safe for $b$. We prove that either \Cref{AlgoBipOn} collects $r_{e_b}$, or there exists some other buyer $b'$ and edge $e'=\{b',i\}$ collected by \Cref{AlgoBipOn} such that $(i)$ $e'=e_{b'}$, $(ii)$ $r_{e'} > r_{e_b}$, and $(iii)$ $e_{b'}$ \emph{cannot} be safe for buyer $b'$. The above three conditions, combined with the fact that, by definition, \Cref{AlgoBipOn} only collects edges in $\Eplus$, and that each buyer can preclude the collection of at most one safe edge, imply the following stronger version of the claimed inequality:
    \begin{align*}
        w(M) \geq \sum_{b\in B} r_{e_b}\1{\text{$e_b$ is safe for $b$}}.
    \end{align*}
    To prove $(i)$, we simply note that, by definition of \Cref{AlgoBipOn}, for each buyer $b'$, either $e_{b'}$ is collected, or \emph{no} edge adjacent to $b'$ is collected.
    Claim $(ii)$ follows by \Cref{def:safeBipartite} and the fact that only edges in $\Eplus$ can be collected by \Cref{AlgoBipOn}. 
    To establish $(iii)$, observe that, by \Cref{def:safeBipartite}, edges $e'$ and $e_b$ cannot be simultaneously safe, as they are incident to the same item $i$.
\end{proof}
We have all the ingredients to prove the main result for bipartite graphs. 

\begin{theorem}\label[theorem]{thm:bipartite}
For the problem of finding a maximum-weight matching in a bipartite graph $\mathcal{G}$ in the online vertex arrival model, \Cref{AlgoBipOn} is $8$-competitive in expectation, i.e., $ 8 \cdot \EXP{w(M)} \geq \EXP{\opt},$ where $\opt$ is the weight of an optimal matching in $\mathcal{G}.$
\end{theorem}
\begin{proof}
    Following the exact same arguments as in \Cref{thm:edgeGen}, $\EXP{w(\Eplus)} \geq \nicefrac{1}{2} \cdot \EXP{w(\Ms)} \geq \nicefrac{1}{4} \cdot \EXP{\opt}$. Finally, combining \Cref{lem:safeRewardBipartite,lem:algRewardBipartite}, we get:
    \begin{align*}
        \EXP{w(M)}
        &\geq \EXP{\sum_{b\in B} r_{e_b} \1{\text{$e_b$ is safe for $b$}}}
        \\
        &\geq \frac{1}{2} \cdot \EXP{\sum_{b\in B} r_{e_b}}= \frac{1}{2} \cdot \EXP{w(\Eplus)}
        \geq \frac{1}{8} \cdot \EXP{\opt}.
    \end{align*}
\end{proof}

\subsection{Truthful Bipartite Matching} \label{sec:truthful}

For the same model as above, with buyers on the left side of $\mathcal{G}$ arriving online and items on the right side available offline, we design an incentive-compatible version of our prophet inequality.
The difficulty is that our $8$-approximation in \Cref{thm:bipartite} relies on each buyer choosing the \emph{largest} incident edge that beats an item's price. 
However, this item need not correspond to the item of highest \emph{utility} to the buyer. Indeed, the price for this item might be much closer to the buyer's value than for other items. In this case, a buyer might prefer to purchase an item which she values significantly less, in order to maximize her utility. For this reason, we cannot expect to show \emph{any} competitive guarantee for the utility of the returned solution.

Clearly, a different approach is needed. We use a pricing-based algorithm that does not rely on restricting the buyer's choice of item (among those with feasible prices that are still available). 
Instead, we ensure a sufficient weight for the resulting matching solely via appropriate thresholds for each buyer and item: this can be seen as doing a worst-case analysis over all threshold-feasible edges instead of choosing a good subset of them, explaining the incurred loss in the approximation. This allows us to make another small change and, among all threshold-feasible items, always assign the \emph{utility-maximizing} one.
In turn, the routine becomes truthful, and we have to mainly care about showing the approximation. Our strategy for setting the aforementioned thresholds can be interpreted as follows.
Assume we take as a starting point not our prophet inequality for the bipartite case, but the one for general graphs with edge arrivals. 
Obviously, arrival orders for the bipartite one-sided model we consider here are highly restricted compared to general edge-arrival. However, assuming the graph is bipartite, each arriving edge consists of exactly one buyer $b$ and one item $i$. Then, \Cref{AlgoOn} assigns prices $p_b$ and $p_i$, where $p_b$ corresponds to the aforementioned buyer's threshold we are planning to use.

Our mechanism works as follows: we determine greedy prices on the sample graph for all vertices just as described above. When a buyer $b$ arrives, we present her not with the greedy price of each item, but instead charge for each item $i$ the maximum over $p_i$ and $p_b$. Then, we let the buyer choose what she likes best, i.e. we assign a utility-maximizing item. Here, $M_S$ is again a greedy matching on the sample graph.
As before (we omit this part of the analysis because it is analogous), one can imagine an offline version of \Cref{AlgoT} that simply draws two values for each edge, and then goes through their set $A$ in non-increasing order, deciding for each value considered whether it belongs to $R$ or $S$ (if this was not yet decided by the twin value for the same edge). We are able to prove the following result for \Cref{AlgoT}: 

\begin{algorithm2e}[!t]
\caption{Truthful Bipartite Prophet Matching}
\label{AlgoT}
\DontPrintSemicolon
Set $M=\emptyset$\;
Compute the greedy matching $\Ms$ on the graph with edge weights according to sample $S$ 
\For{all $e = \{b,i\} \in \Ms$}{
    Set $p_{b}=p_{i}=s_e$
    }
\For{all vertices $k$ not matched in $\Ms$}{
    Set $p_{k}=0$
    }  
\For{each arriving buyer $b \in B$}{
        Set $F(b) = \{e = \{b,i\} \mid r_e \geq \max\{p_{b},p_i\} \text{ and } i \text{ is not matched in } M\}$\;
        Set $\tilde{e} = \{b, i^*\} = \argmax \{ r_e - \max\{p_{b},p_i\} \mid e=\{b, i\} \in F(b)\}$\;
        Add $\tilde{e}$ to $M$ with weight $r_{\tilde{e}}$\;
        Charge buyer $b$ a price of $\max\{p_{b},p_{i^*}\}$
    }
    \Return $M$
\end{algorithm2e}


\begin{theorem}
\label{thm:truthfulBip}
    For the problem of finding a maximum-weight matching in a bipartite graph $\mathcal{G}$ in the online vertex arrival model, \Cref{AlgoT} is truthful and $16$-competitive in expectation, i.e., $ 16 \cdot \EXP{w(M)} \geq \EXP{\opt}$, where $\opt$ is the weight of an optimal matching in $\mathcal{G}.$
\end{theorem}
\begin{proof}
Truthfulness is immediate: on arrival, the unmatched items are offered to the buyer for prices determined independently of the buyer's values, and she is assigned one that is maximizing her utility. By misreporting, her utility can therefore only become worse.

The competitive analysis of the \Cref{AlgoT} follows closely that of our algorithm for the matching problem on general graphs with edge-arrivals, \Cref{AlgoOn}.
Let $\Eprime$ be the set of {\em price-feasible} edges, namely, $\Eprime$ contains any edge $e = \{b,i\}$ such that $r_e \geq \max\{p_b, p_i\}$. Note that, given both in \Cref{AlgoOn} and in \Cref{AlgoT} the prices are computed in the exact same way (via the greedy matching solution using the $S$), $\Eprime$ has the exactly the same meaning in both algorithms. Further, note that $\bigcup_{i \in I} F(i) \subseteq \Eprime$, and that the arrival order, be it edge-arrival or vertex-arrival, has no influence whatsoever on the computation of set $\Eprime$.

Therefore, we can establish, in exactly the same way as in the proof of \Cref{thm:edgeGen}, that 
$$
\mathbb{E}\left[\sum_{b \in B} \valEprime{b}\right] \geq \frac{1}{2} \cdot \mathbb{E}\left[ \sum_{b \in B} \valMs{b}\right] = \frac{1}{2} \cdot \E{w(\Ms)} \geq \frac{1}{4} \cdot \E{\opt},
$$
where by $\val{\hat{E}}{b}$ we denote the maximum-weight over all edges of a weighted set $\hat{E} \subseteq E$ adjacent to buyer $b$.

Let $e_b$ be the maximum-reward edge in $\Eprime$ adjacent to buyer $b$, if such an edge exists. By applying \Cref{def:safe} for the edge-arrival case, we say that an edge $e = \{b,i\}$ is \emph{safe for buyer $b$}, if $e$ is the only edge in $\Eprime$ adjacent to $b$, and there is no edge in $\Eprime$ that is adjacent to item $i$ and has smaller weight than $r_e$. By the exact same arguments as in \Cref{lem:matching:safeprob}, for any buyer $b$ and incident edge $e$ such that the probability of $e_b = e$ is non-zero, we have that $\mathbb{P}[e_b \text{ is safe for }b \mid e_b = e] \geq \frac{1}{4}$. Thus, for any realization of the edge weights in $A$, by the analysis of \Cref{lem:matching:safectononsafe}, we have that
\begin{align*}
    \EXP{\sum_{b\in B} r_{e_b} \1{e_b\text{ is safe for }b}} 
    &\geq  \frac{1}{4}  \cdot \E{\sum_{b \in B} r_{e_b}} = \frac{1}{4}  \cdot \E{\sum_{b \in B} \val{\Eprime}{b}} \geq \frac{1}{16} \cdot \E{\opt}.
\end{align*}
Again, since our proof for the edge-arrival case holds for any (adversarial) arrival order, it can be easily verified that \Cref{AlgoT} is going to collect a reward of at least the weight of the safe edges in $\Eprime$ in expectation. Indeed, by following exactly the arguments of \Cref{lem:matching:safecollect}, we can show that
\begin{align*}
    \E{w(M)} \geq \E{\sum_{b\in B} r_{e_b} \1{e_b\text{ is safe for }b}}.
\end{align*}
The competitive guarantee follows simply by combining the above inequalities.
\end{proof}

\subsection{Transversal Matroids}
\label{sec:transversal}
We now show how our algorithm and analysis for the vertex-arrival model can be easily modified to yield an $8$-competitive algorithm for the case of a {\em transversal matroid}. Recall that, given an undirected bipartite graph $\mathcal{G}=(B \cup I, E)$, a transversal matroid is defined over the ground set of buyers $B$, where a set $S \subseteq B$ is independent (i.e., feasible) if the vertices of $S$ can be matched in $\mathcal{G}$ with a subset of the set of items $I$. Each buyer $b\in B$ is associated with a weight $w_b$ drawn independently from distribution $D_b$, and \emph{every} edge incident to buyer $b$ inherits this weight. We modify our algorithm for bipartite matching, \Cref{AlgoBipOn}, to obtain an analogous \sspi{} for the setting of transversal matroids, \Cref{AlgoTransversalOn}.

\begin{algorithm2e}[!t]
\DontPrintSemicolon
    Set $\Eplus=\emptyset$, $M=\emptyset$\;
    Choose an arbitrary ordering on the items of $I$ \tcp{To break ties within each buyer.} 
    Compute a greedy matching $\Ms$ on the graph with edge weights according to sample $S=\{s_b \mid b\in B\}$, breaking ties for each buyer according to the ordering on $I$\;
    \For{each $e = \{b, i\} \in \Ms$}{
        Set $p_{b}=p_{i}=s_b$
        }
    \For{each vertex $k$ not matched in $\Ms$}{
        Set $p_{k}=0$
        }  
    \For{each arriving buyer $b \in B$}{
        Let $i^*$ be the first item in the ordering on $I$ such that $r_b \geq \max\{p_b,p_{i^*}\}$ and $\{b,i^*\}\in E$\;
        Add $\hat e = \{b,i^*\}$ to $\Eplus$ with weight $r_{\hat{e}}$     \tcp{Only used for the analysis.}
        \If{$i^*$ is not matched in $\M$}{
            Add $\hat e$ to $\M$ with weight $r_{\hat e}$
        }
    }
    \Return $\M$
\caption{Transversal Matroid}
\label{AlgoTransversalOn}
\end{algorithm2e}

In the vertex-arrival model for bipartite graphs, we assumed that the distribution of \emph{all edges} is product form, even for the edges corresponding to a single buyer. We emphasize that it is not clear, in general, how to extend our techniques to allow for \emph{arbitrary} correlation between edges of a single buyer. Indeed, since the definition of the set $A$ in \Cref{AlgoBipOff} crucially assumes that the distribution of \emph{all} edges is product form, and if the edge weights for each buyer were not independent, then the equivalence between the offline and online algorithms proved in \Cref{claim:equivalenceBipartite} no longer holds.
However, our technique extends \emph{directly} to a special type of correlation structure induced by the transversal matroid -- the one where each buyer has the \emph{same} weight for each item in its edge set.  By extending the arguments used to show \Cref{thm:bipartite}, we can establish the following guarantee for this setting:
\begin{theorem}
\label{thm:transversal}
    For the problem of finding a maximum-weight independent set in a transversal matroid $\MM$, \Cref{AlgoTransversalOn} is $8$-competitive in expectation, i.e., $ 8 \cdot \EXP{w(M)} \geq \EXP{\opt}$, where $\opt$ is the weight of a maximum-weight independent set in the transversal matroid $\MM.$
\end{theorem}
\begin{proof}
    The proof follows through minor, though careful, modifications to the arguments used for \Cref{thm:bipartite}. We will briefly sketch the main ideas needed.
    
    One can construct the equivalent offline algorithm in a similar manner as described in \Cref{AlgoBipOff}. In particular, we first order the set of items $I$ in the same manner as in \Cref{AlgoTransversalOn}. Then, we construct a set $A$ of $2|B|$ \emph{buyer} weights (as opposed to $2|E|$ \emph{edge} weights in \Cref{AlgoBipOff}) by adding two values, $a_{b,1},a_{b,2}$ drawn independently from $D_b$ for each buyer $b$ to $A$. Then, we iterate over the set $A$ in decreasing order, and for each value $a$ corresponding to a buyer $b$, we iterate over each edge of the buyer in the predetermined ordering on $I$, adding at most one edge, using the same decision rule as in \Cref{AlgoBipOff}. The equivalence between this offline algorithm and \Cref{AlgoTransversalOn} follows by the same arguments presented in \Cref{claim:equivalenceBipartite}.
    
    Employing the same definition of safe edges as in \Cref{def:safeBipartite}, one can prove, in the exact same manner as in \Cref{lem:lowerBoundProbBipartite}, that for any edge $e$ adjacent to buyer $b$, and denoting $e_b$ as the edge in $\Eplus$ with largest weight,
    \begin{align*}
        \PRO{\text{$e_b$ is safe for $b$} \mid e_b = e} \geq \frac{1}{2},
    \end{align*}
    since the proof of that result relied only on the independence of edge weights for \emph{different} buyers. Just as in \Cref{lem:safeRewardBipartite}, the above result immediately implies that
    \begin{align*}
        \EXP{\sum_{b\in B} r_{e_b} \1{\text{$e_b$ is safe for buyer $b$}}}
        \geq \frac{1}{2} \cdot \EXP{\sum_{b\in B} r_{e_b}}.
    \end{align*}
    Next, to relate the weight of the set collected by the algorithm with the reward of the safe edges, we note that, as in \Cref{lem:algRewardBipartite},
    \begin{align*}
        \EXP{w(M)} \geq \EXP{\sum_{b\in B} r_{e_b} \1{\text{$e_b$ is safe for buyer $b$}}},
    \end{align*}
    since the proof of this result holds for \emph{any} run of the algorithm (not just in expectation), and relies only on the fact that, if an edge $e_b=\{b,i\}$ is safe for buyer $b$, then, by construction, either $e_b$ is added to $M$, or another edge $e'=\{b',i\}$ \emph{that is not safe for buyer $b'$} with reward $r_{e'} \geq r_{e_b}$ is collected instead.
    
    Finally, since this greedy algorithm is still a $2$-approximation to $\opt$, just as in the proof of \Cref{thm:bipartite}, it suffices to show that
    \begin{align*}
        \EXP{\sum_{b\in B} r_{e_b}} \geq \frac{1}{2} \cdot \EXP{w(\Ms)} \geq \frac{1}{4} \cdot \EXP{\opt}.
    \end{align*}
    The argument here follows closely that presented in the proof of \Cref{thm:edgeGen}. Indeed, consider the first time, in a run of the algorithm, when an edge $e = \{b,i\}$ is being processed and $i$ has not yet been matched in $\Ms$. Then, if the coin flip for buyer $b$ lands heads (with probability $\frac{1}{2}$), then the algorithm will add $e$ to $M$, and it must be that the item for buyer $b$ in $\Ms$ (if any such exists) has smaller weight than $r_{e_b}$. Otherwise, if the coin flip lands tails, then $e_b$ is added to $\Ms$ (and no edge for buyer $b$ will be added to $\Eplus$). Thus, denoting $e_b$ as the edge in $\Eplus$ with largest weight, it follows that $\EXP{r_{e_b}}\geq \nicefrac{1}{2} \cdot \EXP{\valMs{b}}$ (where $\valMs{b}$ denotes the weight of the edge adjacent to buyer $b$ in $\Ms$). Applying linearity of expectation, the claimed bound, and thus the theorem, follows. 
    \end{proof}

\section{Budget-Additive Combinatorial Auctions}
\label{sec:budgetadditive}
Here, we consider a special case of combinatorial auctions, where the buyers have budget-additive valuations. Specifically, each buyer $b \in B$ values every item $i \in I$ at some fixed value $v_b(i)$, drawn independently from a distribution for every item-buyer pair. However, once the total reward collected by a buyer $b$ has reached a certain budget $C_b$, additional items have no marginal contribution any more. More formally, 
we define $v_b(I')$ for any $I'\subseteq I$ as
\[
    v_b(I')=\min \left\{ \sum_{i\in I'}v_b(i),\, C_b \right\}.
\]
We assume the budgets $C_b$ to be fixed beforehand and known to the algorithm. Without loss of generality, we further assume that for all $i\in I$ and all $b \in B$, $v_b(i)\leq C_b$ (otherwise, we may pre-process those valuations exceeding $C_b$ to have weight exactly $C_b$). Similarly to \Cref{sec:bipartite}, we assume that buyers arrive in an online adversarial order. Each time a buyer arrives, we observe the reward realizations of the buyer for all the items, and we can assign to her any collection of items that are not already assigned to a different buyer. Equipped with a single sample $s_e$ for each buyer-item edge $e = \{b,i\}$, the goal is to maximize the expected reward collected against that of a prophet who knows all the reward realizations beforehand, and simply chooses the optimal assignment. 

Our algorithm uses the following routine to compute a near-optimal solution $\Gs$ using the samples $S$. It starts by sorting in decreasing order the sample values $S = \{v_{1},v_{2},\dots \}$, then, in that order, each edge is added to the solution $\Gs$ (i.e., an item is assigned to a buyer) whenever the item is available (i.e., has not been previously assigned) and its allocation does not exceed the according buyer's budget. When, for the first time, the algorithm tries to add an item to buyer $b$'s bundle that would exceed her budget, then buyer $b$ becomes {\em blocked} forever and cannot collect any more items (including the current one). We remark that, in the above greedy routine, items are only collected if their marginal contribution equals their value.

Once this greedy solution $\Gs$ is computed on the samples, our algorithm uses it to decide on the online allocation. In particular, we associate to each item $i$ a threshold $\tau_i$ corresponding to the value of the buyer that gets the item in $\Gs$, or to zero if $i$ is unallocated. Every time a new buyer $b$ arrives in the online phase, the items for which $b$'s valuation beats the threshold and which are still not allocated (in the actual solution $M$, not in $\Gs$) are considered in decreasing order of $b$'s weights. Each element is then actually allocated if it fits the buyer's budget in both $M$ and in $\Gs$. While it is clear what we mean by saying that the element fits in $M$, we need to specify what it means to say it also fits in $\Gs$: when a buyer-item pair $e=\{b,i\}$ with value $a$ is considered, which beats $\tau_i$ and fits into the budget of $b$ in $M$, the algorithm computes $C_S(b,a)$, i.e. the total value of the items allocated to $b$ in $\Gs$ restricted to {\em only} items with larger value than $a$. Then the item is actually allocated if and only if $C_S(b,a) + a \le C_b$, and buyer $b$ is not blocked in $\Gs$ after considering all samples larger than $a$. The formal pseudocode is given in \Cref{AlgoBAOn}.

\begin{algorithm2e}[!t]
\DontPrintSemicolon
    Set $\Eprime = \emptyset$, $M=\emptyset$\;
    Set $C_R(b)=0$ for each buyer $b\in B$\;
    Compute the greedy solution $\Gs$ with edge weights according to sample $S=\{s_e \mid e\in E\}$\;
    Let $C_S(b,a)$ denote the total weight assigned to buyer $b$ in $\Gs$, after processing edges of sample-weight larger than $a$, or $C_b$ in the case where buyer $b$ is already blocked in $\Gs$ after considering all edges of sampled weight larger than $a$.\;
    \For{each $e = \{b, i\} \in \Gs$}{
        Set $\tau_{i}=s_e$
    }
    \For{each item $i$ not assigned in $\Gs$}{
        Set $\tau_{i}=0$
    }  
    \For{each arriving buyer $b \in B$}{
        \For{each edge $e=\{b,i\}\in E$ with weight $r_e$, in decreasing order of weight}{
            \uIf{$r_e > \tau_i$ and $r_e + C_S(b,r_e) \leq C_b$}{
                Add $e$ to $\Eprime$ with weight $r_e$ \tcp{Only used for the analysis.}
                \uIf{$r_e + C_R(b) \leq C_b$ and $i$ has not yet been assigned in $M$}{
                    Set $C_R(b) = C_R(b) + r_e$, and add $e$ to $M$ with value $r_e$
                }
            }
        }
    }
    \Return $\M$
\caption{Budget-Additive Combinatorial Auctions}
\label{AlgoBAOn}
\end{algorithm2e}

\paragraph{Equivalent offline algorithm}

Following the usual proof technique, we state in \Cref{AlgoBAOff} an offline procedure that is easier to analyze and is equivalent to the online one. We start by noting that the two algorithms retain the same sets of interest. The proof of this claim follows using similar arguments as in Claims \ref{claim:equivalence} and \ref{claim:equivalenceBipartite}.

\begin{algorithm2e}[!t]
\DontPrintSemicolon
    Set $\Eprime=\emptyset$, $\Gs= \emptyset$, $\M=\emptyset$, $I^S = \emptyset$\;
    Set $C_S(b)=C_R(b)=0$ for each $b\in B$\;
    For each $e\in E$, draw from $D_e$ two values $a_{e,1}$ and $a_{e,2}$\;
    Order $A=\{a_{e,1}, a_{e,2}|e\in E\}$ in a decreasing fashion\;
    \For{each value $a \in A$ in the above order}{
        \uIf{$a$ corresponds to an edge $e=\{b,i\}$ that has never been observed before}{
            Flip a fair coin\;
            \uIf{Heads}{
                Mark $e$ as $R$-used\;
                \uIf{$a + C_S(b) \leq C_b$ and $i \in I^S$}{
                        Add $e$ with value $a$ to $\Eprime$\;
                }
            }\uElse{
                Mark $e$ as $S$-used\;
                \uIf{$i \in I^S$}{
                    \uIf{$a + C_S(b) \leq C_b$}{
                        Set $C_S(b) = C_S(b) + a$, add $e$ to $\Gs$ with value $a$, and remove $i$ from $I^S$\;
                    }\uElse{
                        Set $C_S(b) = C_b$
                    }
                }
            }
        }\uElseIf{$a$ corresponds to an edge $e=\{b,i\}$ which is $R$-used and $i\in I^S$}{
            \uIf{$a + C_S(b) \leq C_b$}{
                Set $C_S(b) = C_S(b) + a$, add $e$ to $\Gs$ with value $a$, and remove $i$ from $I^S$\;
            }\uElse{
                Set $C_S(b)=C_b$
            }
        }
    }
    \For{each $e = \{b,i\}$ in the same order as \Cref{AlgoBAOn}}{
        \uIf{$e \in \Eprime$, $a+C_R(b)\leq C_b$, and $i$ has not been assigned in $M$}{
            Add $e$ to $M$ with value $a$, and set $C_R(b) = C_R(b) + a$\;
        }
    }
    \Return $\M$.\;
\caption{Offline Simulation for Budget-Additive Combinatorial Auctions}
\label{AlgoBAOff}
\end{algorithm2e}

\begin{claim}
\label{claim:equivalentBA}
    The sets $\Eprime$, $\Gs$, and $M$ are distributed the same way when computed by \Cref{AlgoBAOn} as when computed by \Cref{AlgoBAOff}.
\end{claim}
\begin{proof}
By symmetry, the sets of sample and reward values, $S$ and $R$, respectively, follow the same distribution in the online and offline setting.
Further, the construction of the greedy solution $\Gs$ in \Cref{AlgoBAOff} depends \emph{only} on the samples $S$, and thus follows the same distribution as in \Cref{AlgoBAOn}. Since $M$ is constructed in the exact same way by both algorithms using the edges $\Eprime$, it suffices to argue, for each fixed realization of edge weights, that $\Eprime$ follows the same distribution in both algorithms.

Now, observe that the checks performed by \Cref{AlgoBAOff} to add an $R$-used edge $e=\{b,i\}$ to $\Eprime$ exactly correspond to those in \Cref{AlgoBAOn}. Indeed, neither algorithm collects an edge with weight smaller than its sample (guaranteed by the threshold $\tau_i$ in \Cref{AlgoBAOn}, and by the check that $i\in I^S$ in \Cref{AlgoBAOff}). Since \Cref{AlgoBAOff} processes the \emph{edges} in decreasing order of weight, item $i$ is removed from $I^S$ at the time when an edge with weight $a$ is assigned to $i$, which implies that the threshold $\tau_i$ for item $i$ will be set to $a$. Again because of the greedy processing order, the sample capacity $C_S(b)$ when an element $e=\{b,i\}$ with weight $a$ is considered in \Cref{AlgoBAOn} exactly corresponds to $C_S(b,a)$ in \Cref{AlgoBAOff}. Hence, it follows that $\Eprime$ follows the same distribution in both algorithms. 
\end{proof}

\paragraph{Correctness and competitive analysis}
The correctness of \Cref{AlgoBAOn} follows by the fact that the algorithm never assigns a single item to more than one buyer. We focus now on the analysis of the competitive guarantee of our policy. For any set of (weighted) edges $\hat{E}$, we denote by $\hat{E}(b)$ the (weighted) subset of edges in $\hat{E}$ that are incident to buyer $b$. As a first step, we establish that the greedy assignment $\Gs$ is a constant-approximation to the maximum-weight assignment:
\begin{lemma}
\label{lem:greedyBA}
The solution $G_S$ computed on the samples $S$ is a $3$-approximation to the optimal assignment over $S$.
\end{lemma}
\begin{proof}
Let us denote by $\opt(b)$ the items assigned to buyer $b$ in an optimal assignment, and $v^*_b(i)$ the value buyer $b$ realizes from item $i$ in $\opt$. Note that every buyer might (w.l.o.g.) have at most one fractional item $i_b$ in his $\opt$-bundle, where, in that case, $v^*_b(i_b)$ denotes only a fraction of $v_b(i_b)$, accordingly.

We compute our greedy solution $\Gs$ in a way that might miss out on some of the value realized in $\opt$. 
Let's focus on a buyer $b$ and $\opt(b)$. For any $i \in \opt(b)$, one of the following cases is true:
\begin{enumerate}
    \item $i$ is also part of $b$'s greedy bundle $\Gs(b)$.
    \item $i$ is part of some other buyer's $\Gs(b')$, and $v_{b'}(i)\geq v^*_b(i)$.
    \item $i$ is assigned for a lower value than in $\opt$, or not assigned at all.
\end{enumerate}
Let us call items for which these cases hold $I_1,\, I_2$ and $I_3$, accordingly.
The value $\opt(I_1)$ realized on items of the first case is clearly at most that realized by $\Gs(I_1)$, since the marginal contribution of each item in $\Gs$ always equal to its value.
Further, for the same reason as above, together with the fact that each item in $I_2$ gets a better value in $\Gs$, it follows that $\opt(I_2)\leq \Gs(I_2)$.
Finally, in order to conclude the proof, it suffices to bound the loss incurred due to the items in $I_3$. 
If $i\in I_3,\, i\in \opt(b)$, then at the point where value $v_b(i)$ is considered by the greedy, either (i) item $i$ is no longer available, or (ii) buyer $b$ is blocked, or (iii) adding $v_b(i)$ would exceed $C_b$.
Clearly, case (i) cannot be true, since if $i$ had been already assigned by $\Gs$, then because of the decreasing order where the greedy algorithm parses the edge, item $i$ would have been in $I_2$.
Therefore, it has to be that either $b$'s budget was running full, or $b$ was blocked. In both cases, and again due to the decreasing order of consideration, $\Gs(b)\geq \frac{C_b}{2}$.

Now, let us write the value realized by $\opt$ on item set $I'\subseteq I$ as $\opt(I')$. According to the previous considerations, we get
\[
    w(\Gs) \geq \opt(I_1) + \opt(I_2)\text{, and } w(\Gs) \geq \frac 12 \cdot \opt(I_3)
\]
because buyers with \emph{any} $I_3$-items in $\opt(b)$ have at least half-full budget. Together, this yields
\begin{align*}
    \Gs &\geq \max \left\{ (\opt(I_1) + \opt(I_2)),\,\frac 12 \opt(I_3)  \right\} \\
    &\geq \min_{g,f\geq 0 : g+f=1} \max \left\{ g,\, \frac f2 \right\} \cdot \opt =\frac 13 \opt.    
\end{align*}
\end{proof}
In the above algorithm, $\Eprime$ is the (weighted) set of all edges with a reward greater than the according item threshold, which are also considered before the according buyer is blocked in $\Gs$.  Note that similarly to previous proofs, and due to the fact that conditions for adding an element to $\Eprime$ or $\Gs$ are exactly the same, it holds that $\mathbb{E}[w(\Eprime)]$ is at least half of $\mathbb{E}[w(\Gs)]$: this is true because every edge considered for the \emph{first} time lands in both sets equally likely, and an edge considered for the \emph{second} time can only be added to $\Gs$, but then with at most the same value as its first occurrence. As in the previous sections, the proof of our claim relies on constructing a special subset of the edges in $\Eprime$, called the \emph{safe edges}:

\begin{definition}[Safe edges for budget-additive buyers]\label{def:safeBA}
We call an edge $e=\{b,i\}\in E$ \emph{safe for buyer $b$} if the following conditions are true:
\begin{enumerate}
    \item Edge $e$ is in $\Eprime$ with some weight $r_e$.
    \item No edge in $\Eprime$ incident to item $i$ has weight smaller than $r_e$.
    \end{enumerate}
\end{definition}

However, opposed to our earlier proofs, being safe in the sense of belonging to $\Esafe = \{e \in \Eprime \mid \exists~b\in B\text{ s.t. $e$ is safe for $b$}\}$ does not suffice for us to recover an edge in the actual SSPI:
still, it might happen that at the point where buyer $b$ arrives, edge $e=\{b,i\}$ cannot be (fully) realized because $C_b$ is exceeded.
Indeed, we will only be able to extract the value from $\Esafe$ that \emph{does not} exceed the capacity $C_b$ of each buyer $b$.
To bound this, we will have to intersect $\Esafe$ with a (weighted) set $E^+$, which incorporates the capacity constraints for each buyer $b$: 
\[
    E^+(b)=\left\{ a \in \Eprime(b) \mid a + \sum_{a'\in\Eprime(b),\,a'>a} a'\leq C_b\right\}.
\]
We now show a crucial property of $\Eplus$ -- namely, that its expected weight can be lower-bounded using that of the greedy solution, $\Gs$:

\begin{lemma}\label{lem:plusVsGreedyBA}
Let $\Gs$ be the greedy matching on the set of samples. Then,
\begin{align*}
    \EXP{w(\Eplus)} \geq \frac{1}{4} \cdot \EXP{w(\Gs)}.
\end{align*}
\end{lemma}
\begin{proof}
For any fixed value realization $A$ and any buyer $b$, we define 
\[
    X(b)=\left\{ (e,a)\in \Gs(b)\cup\Eprime(b) \mid a \text{ is the larger value drawn for }e\text{ in }A \right\}.
\]
where the tuple $(e,a)$ refers to adding edge $e$ with a weight $w(e)=a$. 
Informally, for each buyer $b$, the set $X(b)$ contains all the values $a_{e,1}$ (i.e., the larger realization) for each edge $e \in \Gs(b) \cup \Eprime(b)$.

Note that deterministically, $w(X) \geq w(\Gs)$, since for every edge $e \in \Gs$ that is weighted with the smaller realization $a_{e,2}$, it has to be that $a_{e,1}$ belongs to $\Eprime$ and, thus, also belongs to $X$. Further, we define
\[
    X^+(b)= \left\{(e, a)\in X(b) \mid a + \sum_{(e,a')\in X(b),\,a'>a}a'\leq C_b\right\}.
\]
as the subset of weighted edges $(e,a)\in X(b)$ such that $a$, together with all weights $a'>a$ in $X(b)$, does not exceed the capacity $C_b$ of buyer $b.$

Given the above definitions, our proof consists of showing that, for any buyer $b \in B$, the following two properties hold:
\begin{itemize}
    \item[(i)] $2 \cdot \mathbb{E}\left[w(X^+(b)) \right]\geq \mathbb{E}\left[w(\Gs(b)) \right] $
    \item[(ii)] $2 \cdot \mathbb{E}\left[ w(E^+(b))\right]\geq  \mathbb{E}\left[w(X^+(b)) \right]$.
\end{itemize}

To show (i), we first rewrite the expected weight of $X^+(b)$, for any buyer $b$, as
\begin{align*}
  \mathbb{E}\left[ w(X^+(b))\right] =~&\PRO{w(X(b))>C_b } \cdot \mathbb{E}\left[w(X^+(b)) \mid w(X(b))>C_b\right]\\
  &+\PRO{ w(X(b))\leq C_b} \cdot \mathbb{E}\left[ w(X^+(b)) \mid w(X(b))\leq C_b \right].  
\end{align*}
Now if at some point, an edge with value $a$ was added to $X(b)$ that made its weight exceed $C_b$, at this point in time it already was true that $w(X^+(b))\geq \frac{C_b}{2}$, due to decreasing order. Thus, it follows that $\mathbb{E}\left[ w(X^+(b))\right]$ is not smaller than the following quantity:
\[
    \PRO{w(X(b))>C_b }\cdot \frac{C_b}{2}
    + \PRO{ w(X(b))\leq C_b} \cdot \mathbb{E}\left[ w(X^+(b)) \mid w(X(b))\leq C_b \right].    
\]
Recall that, independently of the realization of $A$ and the random coin flips, it holds $w(\Gs(b))\leq w(X(b))$. Thus, since the sets $X(b)$ and $X^+(b)$ are identical in the case that $X(b)\leq C_b$, it follows that
\begin{align*}
    \mathbb{E}\left[ w(X^+(b)) \mid w(X(b))\leq C_b \right] 
    &= \mathbb{E}\left[ w(X(b)) \mid w(X(b))\leq C_b \right]\geq \mathbb{E}\left[ w(\Gs(b)) \mid w(X(b))\leq C_b\right].    
\end{align*}
In addition, it holds (since our greedy routine ignores partial items) $w(\Gs(b))\leq C_b$, thus also $\mathbb{E}\left[w(\Gs(b)) \mid w(X(b))>C_b\right]\leq C_b$, and therefore
\begin{align*}
  \mathbb{E}\left[ w(X^+(b))\right] 
  \geq &~\PRO{w(X(b))>C_b } \cdot \frac{1}{2} \cdot \mathbb{E}\left[ w(\Gs(b)) \mid w(X(b))>C_b  \right]\\
  &~+\PRO{ w(X(b))\leq C_b } \cdot \mathbb{E}\left[ w(X^+(b)) \mid w(X(b))\leq C_b \right] \\
  \geq& \frac{1}{2} \cdot \mathbb{E}\left[w(\Gs(b))\right].  
\end{align*}
Now in order to prove statement (ii), we fix some realization of $A$, some value $a\in A$, and some run of the offline algorithm up to the point where $a$ is considered. Note that if $a$ is the weight of the second occurrence of the according edge, it will not be considered for $X(b)$ or $E'(b)$.
However if it is the first and price-feasible, it is added to $X(b)$ with probability $1$ (as it is added to either $\Gs$ or $\Eprime$), and to $\Eprime(b)$ with probability $\frac 12$.
Now, if $(e,a)\in X(b)\cap X^+(b)$, then still, it is added also to $\Eprime(b)$ and therefore, $E^+(b)$, with probability $\frac 12$. Thus, it follows that $\mathbb{E}\left[w(E^+(b))\right]\geq \frac 12 \cdot \mathbb{E}\left[w(X^+(b))\right]$. The proof of the lemma follows by combining inequalities (i) and (ii). 
\end{proof}
Next, we lower bound the probability that an edge $e = \{i,b\} \in \Eplus$ is safe for buyer $b$:

\begin{lemma}\label{lem:safeProbBA}
For any buyer $b \in B$ and edge $e=\{b,i\}$ which is in $\Eplus$ with non-zero probability, we have that $
    \PRO{e \text{ is safe for buyer } b \mid e = \{b,i\} \in \Eplus} \geq \nicefrac{1}{2}.$
\end{lemma}
\begin{proof}
Consider the time when, during a run of \Cref{AlgoBAOff}, $e$ is added to $\Eplus$.
By construction, the algorithm only adds edges which are $R$-used and are being processed for the first time. Now, consider the next time when an edge $e'=\{b',i\}$ is being considered by the algorithm with weight $a$ which allows it to be allocated in $\Gs$. Now, if $e'$ is being processed for the second time, then it must be $R$-used and added to $\Gs$ (and thus, at this point, $i$ is removed from $I^S$). Thus, in this case, no other edge $e''=\{b'',i\}$ with value smaller than $r_e$ could be added to $\Eplus$, guaranteeing that both conditions of \Cref{def:safeBA} are satisfied. In the case when $e'$ is processed for the first time, then, with probability $\nicefrac{1}{2},$ $e$ is $S$-used and added to $\Gs$, again guaranteeing that no other element of smaller weight is added to $\Eplus,$ and thus that the conditions of \Cref{def:safeBA} are satisfied.
\end{proof}
In the following lemmas, we relate the expected weight of $M$ with that of $E^+$.

\begin{lemma}\label{lem:safeAndPlusBA}
Recall that $\Esafe$ is the set of edges $e=\{b,i\}\in \Eprime$ which are safe for some buyer $b$. Then, $\EXP{w(\Esafe \cap \Eplus)} \geq \frac{1}{2} \cdot \EXP{w(\Eplus)}.$
\end{lemma}
\begin{proof}
Here, the argument is essentially the same as in \Cref{lem:matching:safectononsafe}, replacing the use of \Cref{lem:matching:safeprob} with \Cref{lem:safeProbBA}. In particular, note that, taking (without loss of generality) $a_{e,1}$ to be the larger realization of the edge weight for $e$,
\begin{align*}
    \EXP{w(\Esafe \cap \Eplus)}
    &=  \EXP{\sum_{e=\{b,i\}\in E^+} r_e \cdot \1{\text{$e$ is safe for $b$}}} 
    \\
    &= \sum_{e=\{b,i\}\in E} a_{e,1} \cdot \PRO{\text{$e$ is safe for $b$, $e\in\Eplus$}}.
\end{align*}
Thus, by \Cref{lem:safeProbBA},
\begin{align*}
    \PRO{\text{$e$ is safe for $b$, $e\in\Eplus$}}
    &= \PRO{\text{$e$ is safe for $b$} \mid e\in\Eplus} \PRO{e\in\Eplus}\\
    &\geq \frac{1}{2} \cdot \PRO{e\in \Eplus}.
\end{align*}
By combining these two results, we conclude that
\begin{align*}
    \EXP{w(\Esafe \cap \Eplus)}
    \geq \frac{1}{2} \sum_{e=\{b,i\}\in E} a_{e,1} \cdot \PRO{e\in \Eplus}
    &= \EXP{w(\Eplus)},
\end{align*}
as desired.
\end{proof}

\begin{lemma}\label{lem:safeVsMBA}
Let $M$ be the assignment obtained in a run of \Cref{AlgoBAOff}. Then, it holds that $
    \EXP{w(M)}
    \geq \EXP{w(\Esafe \cap \Eplus)}.$
\end{lemma}
\begin{proof}
We prove a stronger result -- namely, for \emph{any} run of \Cref{AlgoBAOff}, $w(M) \geq w(\Esafe \cap \Eplus)$.
Consider any edge $e=\{b,i\} \in \Eplus$ which is safe for buyer $b$. By construction of the set $\Eplus$, and since edges in $\Eprime$ are processed in decreasing order of weight, when $e$ is processed by the algorithm, it must be the case that either $e$ is added to $M$, or some other edge $e'=\{b',i\}\in\Eprime$ was already chosen incident to item $i$ before edge $e$ was processed during the online phase (since the algorithm only accepts items in the set $\Eprime$). 
We emphasize here that by definition and due to the weight-decreasing order, an edge $e=\{b,i\}\in \Eplus$ can \emph{never} be rejected because it exceeds the buyer's capacity $C_b$. 
Since $e$ is safe for buyer $b$, and since \Cref{AlgoBAOff} \emph{never} collects truncated rewards by construction, it must be the case that $r_{e'} > r_e$. Therefore, $e'$ cannot be safe for buyer $b'$, since the second condition \Cref{def:safeBA} is violated by edge $e$. Hence, for any run of \Cref{AlgoBAOff} for which $e$ is safe for buyer $b$, an element of weight at least $r_e$ is collected in $M$. By aggregating this result for all edges in $\Eplus$, we reach the desired result. 
\end{proof}
We are ready for the main theorem in the budget-additive setting. 

\begin{theorem}\label[theorem]{thm:budgetAdditive}
For the problem of finding a maximum-weight assignment in a budget-additive combinatorial auction in the online vertex arrival model, \Cref{AlgoBAOn} is $24$-competitive in expectation, i.e., $
    24 \cdot \EXP{w(M)} \geq \EXP{\opt},$
where $\opt$ is the weight of a maximum-weight assignment.
\end{theorem}
\begin{proof}
Combining the results of \Cref{lem:safeVsMBA} and \Cref{lem:safeAndPlusBA}, we can relate the expected reward collected by \Cref{AlgoBAOn} with that collected by $\Eplus$:
\begin{align*}
    \EXP{w(M)} \geq \EXP{w(\Esafe \cap \Eplus)} \geq \frac{1}{2} \cdot \EXP{w(\Eplus)}.
\end{align*}
At this point, we can combine the result of \Cref{lem:plusVsGreedyBA} (that allows us to relate the expected weights of $\Eplus$ and the greedy solution with respect to the samples, $\Gs$) with \Cref{lem:greedyBA} (that relates the expected weights of the greedy and optimal solutions) to obtain that
\begin{align*}
    \EXP{w(\Eplus)} \geq \frac{1}{4} \cdot \EXP{w(\Gs)} \geq \frac{1}{12} \cdot \EXP{\opt}.
\end{align*}
Combining these results, we conclude that
\begin{align*}
    \EXP{w(M)} \geq \frac{1}{24} \cdot \EXP{\opt}.
\end{align*}
\end{proof}


\section{From \texorpdfstring{$\alpha$}{a}-Partition to Single-Sample Prophet Inequalities}

\label{sec:reduction}
In this section, we turn our attention to the \sspi{} problem under matroid feasibility constraints. We show how the recent work of \citet{RubinsteinWW20} can be applied on an interesting class of matroids that satisfy a certain property, called $\alpha$-partition. For this class, we provide improved competitive guarantees (essentially by a factor of $2$) for the \sspi{} problem, comparing to those following from the reduction to \oos \cite{AzarKW14}. 

We refer to a matroid $\Mat=(E,\I)$ over a ground set $E$ as a {\em simple} partition matroid\footnote{We use the term ``simple'' in order to distinguish the partition matroid from its common definition, where, from each set $E_l$ of the partition, more than one element might be collected.}, if there exists some partition $\bigcup_{l \in [k]} E_l $ of $E$ such that $I \in \I$ if and only if $|I \cap E_l| \leq 1$ for each $l \in [k]$. In other words, a set is independent only if it contains at most one element from each set of the partition.

We consider the following property of several matroids, called $\alpha$-partition (slightly adapted from~\citet{BDGIT09}):

\begin{definition}[$\alpha$-Partition Property] \label{def:alphapartition}
A matroid $\Mat=(E,\I)$ satisfies an $\alpha$-partition property for some $\alpha \geq 1$ if for any weight vector $w$ on the elements, after sequentially observing the weight of a (possibly random) subset $S \subset E$ of the elements selected independently of the weights, one can define a simple partition matroid $\Mat' = (E' , \I')$ on a ground set $E' \subseteq E \setminus S$, such that
\begin{align*}
    \EXP{\max_{I' \in \I'} w(I')} \geq \frac{1}{\alpha} \cdot \max_{I \in \I} w(I)
    \qquad \text{and} \qquad
    \I' \subseteq \I,
\end{align*}
where the expectation is taken over any randomness in the partitioning procedure.
\end{definition}

In fact the above definition is a slight adaptation of one in \citep{BDGIT09}. 
Specifically, we additionally capture the case where the transformation may observe the weight of a subset $S \subset E$ of the elements, and that this $S$ is never included in the ground set of the produced partition matroid. Note that the above definition also permits transformations that do not depend on the weights of \emph{any} sample (i.e., $S=\emptyset$), or that are deterministic. We prove the following meta-theorem for any matroid with an $\alpha$-partition property:

\begin{theorem}
\label[theorem]{thm:reduction}
For any matroid $\Mat$ that satisfies an $\alpha$-partition property for some $\alpha \geq 1$, there exists a ${2\alpha}$-competitive policy for the corresponding \sspi~problem. Further, if the $\alpha$-partitioning can be performed in polynomial time, then the policy is also efficient.
\end{theorem}

\begin{proof}

Given any matroid $\Mat = (E,\I)$ which satisfies an $\alpha$-partition property,
we describe how to construct a simple $2\alpha$-competitive policy for the \sspi{} on $\Mat$. The policy proceeds in two phases: (Offline phase) We construct a simple partition matroid $\Mat'=(E',\I')$ from $\Mat$, using its $\alpha$-partition property. For each element $e \in S\subset E$ that needs to be observed by the transformation, we feed the corresponding sample $s_e$. Let $\bigcup_{l \in [k]} E_l = E'$ be the constructed partition. For each group $E_l$ of $\Mat'$, we set a threshold $\tau_l = \max_{e \in E_l} s_e$, equal to the value of the largest sample of the elements included in $E_l$. Finally, we initialize $I=\emptyset$. (Online phase) For each element $e\in E$ arriving in adversarial order, we immediately reject it (without even observing the associated reward), if $e \notin E'$. Otherwise, assuming that $e\in E_l$ for some $l \in [k]$, we accept $e$ and add it to $I$ if and only if (i) $I \cap E_l = \emptyset$ (i.e., no other element in $E_l$ has been accepted so far) and (ii) the reward satisfies $r_e > \tau_l$. 

We now show that the above policy is $2\alpha$-competitive for the \sspi{} problem on $\Mat$. First observe that, by construction of the policy, the set of elements collected satisfies $I \in \I'$ and, thus, $I \in \I$ by definition of the $\alpha$-partition property. In order to establish the competitiveness of the policy, we first note that the maximum-reward independent set of $\Mat'$ is simply the collection of maximum-reward elements in each group of the partition. Note further that the online phase of our policy consists of running several instances of the algorithm of \cite{RubinsteinWW20}, one for each partition. Therefore, by linearity of expectation, our policy collects in expectation a $\nicefrac{1}{2}$-fraction of the expected maximum-reward within each partition and, hence, is $2$-competitive against the optimal policy for $\M'$.  
Finally, given that the partitioning is performed using the samples, the expected maximum-reward independent set of $\Mat'$ has the same expectation as if the partitioning was performed using the rewards, since the samples and rewards are identically distributed. By applying \Cref{def:alphapartition}, we can conclude that our policy is $2\alpha$-competitive for $\Mat.$ 
\end{proof}

\begin{remark}
We remark that \Cref{thm:reduction} can be extended in the case of other combinatorial problems (not necessarily matroids) that satisfy a partitioning property, similar to \Cref{def:alphapartition}.
\end{remark}

As an application of \Cref{thm:reduction}, we obtain improved \sspi{}s for a number of settings considered in \citet{AzarKW14}, displayed in last four entries in \Cref{table:main}. Indeed, the improvement in these results comes by noticing that each of these matroid environments satisfy the $\alpha$-partition property for constant $\alpha$, and rely on a $4$-competitive single-choice \oos{}. As a result of \Cref{thm:reduction} and \citet{RubinsteinWW20}, we can replace this \oos{} with a $2$-competitive \sspi{}, thus improving the competitive guarantee by a factor of $2$.

\section{From Order-Oblivious Secretaries to Pointwise-\sspi{}s}
\label{sec:sspitooos}

Since \emph{direct} \sspi{}s not only appear to offer better approximation than \oos{} policies, but also obviously have access to more information (the samples) comparing to an \oos{} algorithm, one can hope to derive constant-factor \sspi{}s even where according \oos{} policies fail to exist. 
For example, a major open problem in optimal stopping theory is the existence of 
$\mathcal{O}(1)$-competitive policies for the matroid secretary problem.
While Babaioff, Immorlica, and Kleinberg \citep{BIK07}
conjectured the existence of such a policy, and $\mathcal{O}(1)$-competitve
policies are known for many special cases \citep[see, e.g., ][]{STV21},
the state-of-the-art for general matroids is a $\O(\log\log(\mathrm{rank}))$-competitive policy \cite{FSZ14,Lachish14}. In stark contrast, \citet{KleinbergW12} give an \emph{optimal} $2$-competitive policy for the related matroid
prophet inequality problem under full distributional knowledge. 
Thus, given that the existence of an $\mathcal{O}(1)$-competitive policy for the \sspi{} problem on general matroids is a persistent open question, it is natural to ask whether the problem exhibits similarities in terms of hardness with its (order-oblivious) secretary counterpart. In this section, we make progress towards the answer to this question, by relating the OOS problem with a wide class of \sspi{} policies, which we refer to as ``pointwise''-SSPIs (or ``\psspi{}s''). 

\begin{definition}[\psspi{}]
An \sspi{} policy is called \psspi{}, if it maintains its competitive guarantee when the reward and sample of each element $e$ are generated as follows: given two arbitrary non-negative weights $a_{e,1}$ and $a_{e,2}$, we flip an independent fair coin to decide whether $r_e = a_{e,1}$ and $s_{e} = a_{e,2}$, or the opposite.
\end{definition}

Note that, as opposed to the \sspi{} problem, where each reward and sample are drawn independently from the same distribution, 
\psspi{} allows
for the reward and sample to be \emph{correlated}. Further, by the exact same reduction (and proof) presented in \cite{AzarKW14}, which reduces \sspi{} to \oos{}, one can in fact show the following stronger result:

\begin{theorem}[\citet{AzarKW14}] \label[theorem]{lem:OOStoPSSPI} Given any $\alpha$-competitive \oos{}, one can construct an $\alpha$-competitive \psspi{}.
\end{theorem}

An important observation is that every known \sspi{} policy is \emph{also} a \psspi{} policy with the same competitive guarantee. 

\begin{observation}\label[observation]{obs:psspi}
Every known \sspi{} policy, including the single-choice \sspi{} of \citet{RubinsteinWW20}, any \sspi{} in \citet{AzarKW14}, and our policies, falls into the class of \psspi{}s.
\end{observation}

Perhaps surprisingly, we can establish the following result, which is a partial converse to the reduction in \Cref{lem:OOStoPSSPI}: 

\begin{theorem}\label[theorem]{thm:oosAndPsspiEquivalent}
For any downward-closed feasible set of elements, given any $\alpha$-competitive \psspi{}, one can construct a $2\alpha$-competitive \oos{}. 
\end{theorem}
\begin{proof}
Given an $\alpha$-competitive policy $\mathcal{P}$ for the \psspi{} problem, we construct a policy $\mathcal{A}$ for the \oos{} problem as follows:
\paragraph{Phase 1} Let $k \sim \mathrm{Binomial}(n, \nicefrac{1}{2})$
be a random number of elements. Observe (without collecting) the rewards of the first $k$ elements in uniformly random order, and feed them to $\mathcal{P}$ as the \emph{samples} of the corresponding elements. For the elements that are not yet observed, we feed to $\mathcal{P}$ \emph{zeros} as the corresponding samples. 
\paragraph{Phase 2} For the remaining elements arriving in adversarial order, the \oos{} policy  $\mathcal{A}$ mimics the decisions of $\mathcal{P}$ (i.e., the \oos{} policy accepts if and only if $\mathcal{P}$ accepts), skipping the elements that are parsed in Phase 1.

We now show that the above policy is $2\alpha$-competitive for the \oos{} problem. Let $\opt(\text{\oos{}})$ be the optimal reward for the \oos{} problem on some instance, and let $\opt(\text{\psspi{}})$ be the prophet's reward on the corresponding instance where the rewards of $k$ elements chosen uniformly at random are set to zero.
For our selection of $k$, it is not difficult to verify that the reward of each element is zero (and its weight is provided to $\mathcal{P}$ as a sample) with probability half.
Hence, given that the optimal solution to \oos{} is feasible for \psspi{} and the reward of each element ``survives'' with probability half, it follows that
\[
    \EXP{\opt(\text{\psspi{}})} \geq \frac{1}{2} \cdot \opt(\text{\oos{}}).
\]
Let $w(\mathcal{P})$ and $w(\mathcal{A})$ be the reward collected by policy $\mathcal{P}$ and $\mathcal{A}$, respectively, in the above reduction. Note that by construction of our reduction, every element that is parsed in Phase 1 of $\mathcal{A}$, has zero reward and, thus, is never collected by $\mathcal{P}$ in the online phase, without loss of generality. Given that these elements are skipped by $\mathcal{A}$ in Phase 2 and, for the rest of the elements, $\mathcal{A}$ mimics the decisions of $\mathcal{P}$, we can see that $w(P) = w(A)$. Finally, since $\mathcal{P}$ is an $\alpha$-competitive \psspi{}, we can conclude that
\begin{align*}
    \EXP{w(\mathcal{A})} =     \EXP{w(\mathcal{P})} \geq \frac{1}{\alpha} \cdot \EXP{\opt(\text{\psspi{}})} \geq \frac{1}{2\alpha} \cdot \opt(\text{\oos{}}),
\end{align*}
thus showing that $\mathcal{A}$ is a $2\alpha$-competitive \oos{} policy. We remark that the downward-closedness assumption on the feasible set is only used for guaranteeing that rejecting elements in Phase 1 of $\mathcal{A}$ cannot lead to an infeasible solution. 
\end{proof}
An immediate corollary of \Cref{obs:psspi} and \Cref{thm:oosAndPsspiEquivalent} is that any \emph{known} \sspi{} policy effectively provides an algorithm for the corresponding (order-oblivious) secretary problem. By applying \Cref{thm:oosAndPsspiEquivalent} to \Cref{thm:budgetAdditive}, we obtain a $48$-competitive \oos{} policy for budget-additive combinatorial auctions. To the best of our knowledge, this is the first order-oblivious secretary policy for this problem. By \Cref{thm:oosAndPsspiEquivalent}, solving the \psspi{} problem on a downward-closed feasible set (including general matroids) \emph{immediately implies} an \oos{} policy on the same feasible set.  Hence, unless there exists a $\mathcal{O}(1)$-competitive \sspi{} which is \emph{not} a \psspi{}, the \sspi{} problem is at least as \emph{as hard} as the \oos{} problem (up to constant factors)! Since all currently known \sspi{} are \psspi{}, it is natural to ask whether this observation can be generalized:

{\em Given any $\alpha$-competitive \sspi{} policy, does there also exist an $\mathcal{O}( \alpha)$-competitive \psspi{} policy? Or, stated differently, are the classes of \sspi{} and \psspi{} equivalent up to constant factors?}

In case of a positive answer to the previous problem,
solving the constant-factor \sspi{} problem on general matroids would imply also resolving the matroid secretary conjecture of \citet{BIK07}. In any case, and more generally: since \psspi{}s provide constant-factor approximations to the same class of problems as \oos{}s do, only \emph{non-}\psspi{} might be more powerful, i.e., only with such policies, the \sspi{} paradigm can ever cover a greater scope of settings than \oos{}.

\section*{Acknowledgments}
This work was supported by the ERC Advanced Grant 788893 AMDROMA ``Algorithmic and Mechanism Design Research in Online Markets'', the MIUR PRIN project ALGADIMAR ``Algorithms, Games, and Digital Markets'', and partially funded by NSF grant 2019844.

\bibliographystyle{plainnat}
\bibliography{references}
\end{document}